\documentclass{IEEEtran}

\usepackage{ifpdf}
\ifCLASSINFOpdf
\usepackage[pdftex]{graphicx}
\graphicspath{{./Figures/}}
\DeclareGraphicsExtensions{.pdf,.jpeg,.png}
\else
\usepackage[dvips]{graphicx}
\DeclareGraphicsExtensions{.pdf}
\usepackage[caption=false, font=footnotesize]{subfig}
\fi

\usepackage{tabularx, booktabs}
\usepackage{supertabular}
\usepackage{amssymb}
\usepackage{amsthm}
\usepackage[ruled,vlined,linesnumbered]{algorithm2e}
\usepackage{aligned-overset}
\usepackage{array}
\usepackage{balance}
\usepackage{cite}
\usepackage{xcolor}
\usepackage{epstopdf}
\usepackage{enumitem}
\usepackage{mathtools}
\usepackage{makecell}
\usepackage{multicol}
\usepackage{multirow}
\usepackage[caption=false,font=small]{subfig}
\usepackage{textcomp}
\usepackage{stfloats}
\usepackage{url}
\usepackage{verbatim}
\usepackage{siunitx}
\usepackage{algorithmic}
\usepackage{cases}

\newtheorem{proposition}{Proposition}

\newtheorem{remark}{Remark}

\begin{document}
\title{\color{black} Probabilistic Denoising-Enhanced ISAC for Stochastic Cluttered Mobile Environments}
\author{
{Nghia~Thinh~Nguyen~and~Tri~Nhu~Do}

\thanks{N. T. Nguyen and T. N. Do are with the Department of Electrical Engineering, Polytechnique Montr\'{e}al, Montreal, Quebec, Canada. Emails: nghia-thinh.nguyen@etud.polymtl.ca, tri-nhu.do@polymtl.ca}
\thanks{Codes are available at \cite{github_repo}: \texttt{https://github.com/TND-Lab/PDISAC}}
}

\maketitle

\color{black}
\begin{abstract}
In this paper, we propose Probabilistic Denoising ISAC (PDISAC), a framework built on a multi-bit slot-partitioned ISAC waveform: by partitioning each maximal-length sequence into alternating pilot and data slots, we embed multiple bits per sequence through symbol-level spreading, multiplying the data rate while every chip retains the deterministic radar code. The added throughput, however, injects data-dependent, non-white sidelobes into the range-Doppler (RD) heatmap that degrade matched-filter (MF) sensing. Rather than modifying the MF receiver, we develop RDPDNet, a lightweight probabilistic denoising network inserted between RD-map formation and constant-false-alarm-rate detection; training it with an adversarial frequency-mixup mechanism, we suppress the data-induced sidelobes and thermal noise without knowledge of the embedded symbols. We further characterize the statistics of the geometry-determined channel. The fundamental performance limits of the design are then analyzed through an analytical lower bound, a semi-analytical bit error rate (BER), and an average capacity that tie the slot allocation and sequence length to the sensing-communication trade-off. Through analytical and numerical results over a realistic urban geometry, we show that RDPDNet absorbs most of the data-embedding sensing penalty and markedly lowers the RMSE at low SNR, while the conventional data-free chain attains the bias-adjusted benchmark at high SNR. Moreover, increasing the slot allocation raises the data rate at the expense of a higher BER, exposing a tunable sensing--communication trade-off.
\end{abstract}

\section{Introduction}

Integrated sensing and communication (ISAC) is a key enabler of future wireless networks, letting a common waveform and shared front-end support radar sensing and data communication simultaneously \cite{liu2022integrated}. Operation in stochastic, cluttered, and mobile environments is challenging, since time-varying propagation, multipath clutter, and multiple moving targets distort the sensing returns while impairing the communication link \cite{Jiang2025}. We consider a monostatic ISAC system whose co-located transmitter and receiver share a common waveform, enabling multi-target range-Doppler estimation from backscattered echoes while concurrently serving a mobile UE.
Among candidate ISAC waveforms, we adopt a phase-modulated continuous-wave (PMCW) design \cite{Park2025Compressive}.  This radar-centric waveform transmits a constant-envelope, phase-coded sequence with excellent periodic autocorrelation, low peak-to-average power ratio, and a favorable ambiguity function \cite{Giroto2022JRC}. These properties make it robust for range--Doppler (RD) sensing while remaining compatible with simple phase modulation for data embedding.

The PMCW receiver chain comprises matched filtering, slow-time Doppler processing, RD map formation, clutter suppression, detection, and parameter estimation \cite{Jin2019}. We focus on enhancing the RD heatmap, the two-dimensional range-Doppler image from which targets are detected and estimated. In stochastic, cluttered, and mobile environments it is corrupted by thermal noise, nonstationary clutter \cite{Ren2023}, time-varying multipath \cite{Tagliaferri2024}, correlated interference, and data-dependent sidelobes; as the user, targets, and scatterers move, these disturbances shift over time, producing Doppler spreading, peak migration, spurious peaks, and unpredictable noise. Methods calibrated for a fixed environment thus face distribution shifts and may miss weak targets or mistake clutter for detections \cite{Yan2023Innovative}. The heatmap is therefore a bottleneck of the sensing chain, translating directly into missed targets, false alarms, and biased estimates \cite{Cao2025}, so we single it out.

Two families of methods enhance the RD heatmap. Conventional, rule-based filters \cite{fuhrmann1992cfar} (matched filtering, windowing, adaptive clutter rejection) maximize the signal-to-noise ratio (SNR) and are computationally light, but assume fixed waveform parameters and well-characterized noise, and degrade under multi-target \cite{Xu2023Superimposed}, time-varying, and structured non-white disturbance \cite{Yan2023Innovative}. Artificial-intelligence (AI)-based denoisers instead learn to clean the heatmap directly from data \cite{Roldan2023, Wang2025GAN}, adapting to complex environments, but are typically appended as post-processing, add latency, and are evaluated only numerically without closed-form guarantees.

Within the AI methods, probabilistic machine learning (ProbML) is especially promising: via variational-autoencoder (VAE) and hierarchical-Bayesian latent models \cite{CCECE, Feintuch2023}, it treats the RD map as a distribution and infers a latent representation of the clean response, capturing structured, non-white disturbance with calibrated uncertainty. Yet existing ProbML denoisers are applied generically, decoupled from the waveform parameters that generate the artifacts, and lack closed-form links to fundamental limits. Against the closest denoisers, i.e., the generative adversarial network (GAN)-based RD cleaner \cite{Wang2025GAN} and the VAE/hierarchical-Bayesian denoisers \cite{CCECE, Feintuch2023}, our proposed approach, as detailed below, differs in two respects: it is jointly designed with the multiple data symbols embedded into ISAC waveform that generates the sidelobe artifact; it is tied to a closed-form Cram\'{e}r-Rao lower bound (CRLB) and a semi-analytical bit error rate (BER), which those works do not provide.\label{anchor_comment_2_m11}\label{anchor_comment_2_m14}\label{anchor_comment_3_6}\label{anchor_comment_E_10}\label{anchor_comment_E_2b}\label{anchor_comment_2_2b}\label{anchor_comment_4_9b}

The unaddressed gap is that these methods enhance the RD heatmap under fixed waveform settings, ignoring how the waveform parameters shape the disturbance to be removed \cite{Li2025}. In PMCW the two time scales couple the functions: the fast-time chip sequence sets the range resolution, while the slow-time domain yields the Doppler estimate and carries data by sign-modulating the code across repetitions \cite{Foroozmehr2024}. Embedding more data symbols raises the rate but perturbs the periodic autocorrelation, injecting fluctuating sidelobes, i.e., spurious or ambiguous target peaks, and exposing a sensing-communication trade-off controlled by the sequence length and the symbol-to-slot allocation. This trade-off is acute for radar-centric PMCW, whose deterministic code carries fewer symbols than communication-centric multicarrier waveforms in ISAC. Thus, fixed-parameter designs select a single operating point \cite{Li2025LowRange}. Because the induced sidelobes are data-dependent and non-white, a probabilistic stage designed jointly with the slot allocation can remove them; we exploit this by parameterizing the sounding-to-data allocation and absorbing the resulting disturbance, relaxing the trade-off rather than balancing it at a fixed operating point.

In this paper, we propose Probabilistic Denoising ISAC (PDISAC), a multi-bit, slot-partitioned ISAC waveform jointly designed with a probabilistic range–Doppler heatmap enhancement stage. Each code is split into alternating sounding and data slots, exposing the number of slow-time data symbols as an explicit design parameter, and a lightweight RD Probabilistic Denoising Network (RDPDNet), trained end-to-end with an adversarial frequency-mixup (AFM) mechanism, cleans the resulting heatmap by separating true peaks from data-induced sidelobes and noise. Our contributions are as follows.

\begin{itemize}
	\item We propose a multi-bit slot-partitioned ISAC waveform based on a maximum length sequence (MLS), enabling multiple communication bits per sounding sequence through symbol-level spreading while maintaining deterministic code properties. The proposed design provides explicit control over the throughput-reliability-sensing trade-off compared with conventional one-bit-per-code MLS-ISAC.
	
	\item We characterize a geometry-based channel model, where each LOS/NLOS path is deterministic for a given scene, while the aggregate NLOS component across random scenes admits a central-limit characterization with Rayleigh/Rician envelopes without imposing per-path fading assumptions. The model is validated through Monte-Carlo (MC) simulations.
	
	\item We develop RDPDNet, a lightweight hierarchical latent-variable denoiser that formulates RD-heatmap enhancement as Bayesian posterior inference and suppresses data-induced interference without knowledge of the embedded symbols. The proposed network is trained with an AFM loss to improve robustness against worst-case interference. This work introduces an RD denoiser co-designed with the waveform parameters responsible for the induced sidelobe artifacts.
	
	\item We derive closed-form performance limits, including the CRLB for range and velocity, a semi-analytical BER, and an average capacity, establishing the sensing-communication trade-off of the proposed waveform design. Extensive evaluations over realistic urban geometries demonstrate that RDPDNet mitigates the sensing degradation caused by data embedding and improves low-SNR RMSE while remaining consistent with the bias-adjusted CRLB at high SNR.
\end{itemize}

{\color{black}
\noindent\textbf{Notation.} We use $x$ and $X$ to denote a signal and its Fourier transform, respectively. Vectors and matrices are represented by $\vec{x}$ and $\mathbf{X}$ for specific realizations. \label{anchor_comment_4_9}\label{anchor_comment_E_10b}
}

\section{ISAC System and Signal Model} 

We consider a the ISAC system consisting of a transmitter (Tx), a receiver (Rx), a set of multiple moving targets, a set of multiple stationary scatterers, and a single user equipment (UE), as illustrated in Fig.~\ref{fig_topology}. 
The transmitter is equipped with single antenna, transmit power $P_{\rm tx}$, and antenna gain $G_{\rm tx}$. Assuming the transmitter is stationary at a position $\vec{l}_{\rm tx} \in \mathbb{R}^{3 \times 1}$.  The co-located receiver has single antenna, with power gain $G_{\rm rx}$, sharing the same position with the transmitter $\vec{l}_{\rm rx} = \vec{l}_{\rm tx}$.
The moving targets consists of $N_{\rm tars}$ targets whose three-dimensional positions, velocities, and radar cross-sections~(RCS) are collected in the matrices $\mathbf{L}_{\rm tars} \in \mathbb{R}^{3 \times N_{\rm tars}}$, $\mathbf{V}_{\rm tars} \in \mathbb{R}^{3 \times N_{\rm tars}}$, and $\vec{\sigma}_{\rm tars} \in \mathbb{R}^{1 \times N_{\rm tars}}$, respectively.
Similarly, the stationary scatterers consists of $N_{\rm scats}$ scatterers characterized by their positions $\mathbf{L}_{\rm scats} \in \mathbb{R}^{3 \times N_{\rm scats}}$, velocities $\mathbf{V}_{\rm scats} \in \mathbb{R}^{3 \times N_{\rm scats}}$,
and RCS values $\vec{\sigma}_{\rm scats} \in \mathbb{R}^{1 \times N_{\rm scats}}$.
The UE is considered a moving target, equipped with single antenna with its gain $G_{\rm ue}$. The UE position and velocity are denoted by $\vec{l}_{\rm ue} \in \mathbb{R}^{3 \times 1}$ and $\vec{v}_{\rm ue} \in \mathbb{R}^{3 \times 1}$, respectively. 

\begin{figure}[!t]
    \centering
    \color{black}
    \includegraphics[width=\linewidth]{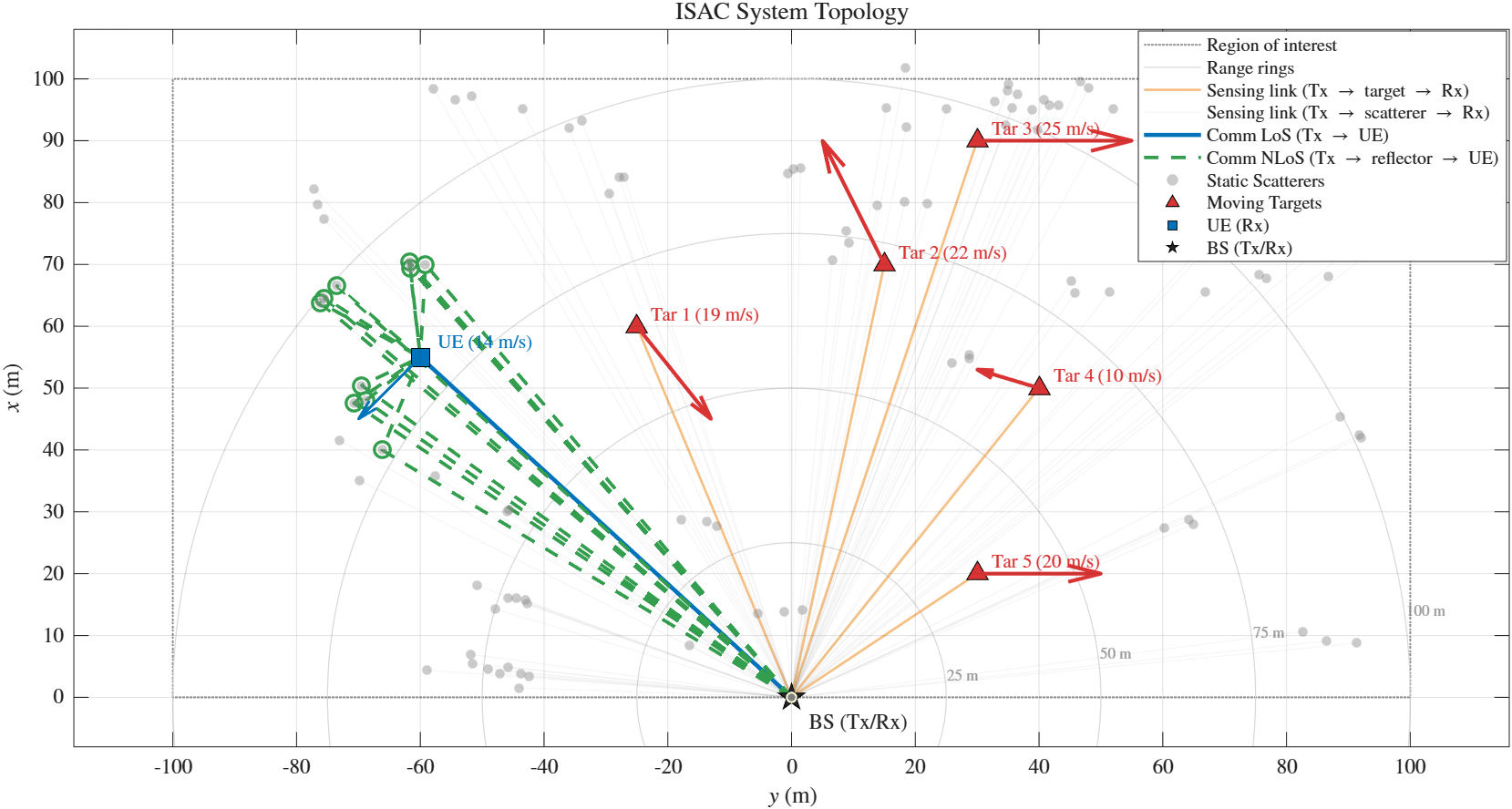}
    \caption{One realization of the considered stochastic cluttered environment, showing the co-located Tx/Rx, UE, moving targets, and stationary scatterers.}
    \label{fig_topology}
\end{figure}

According to \cite{liu2022integrated}, the data of an ISAC system is transmitted using a dedicated waveform design. In this work, we employ an extended maximal-length sequence (MLS). The MLS is used as the sounding sequence and implemented via a Pseudorandom Binary Sequence (PRBS), denoted as $\vec{p}_{\rm prbs} \in \{0, 1\}^{N_{\rm chip} \times 1}$. Here, we extended it to $N_{\rm chip} = 2^m$ to become an extended MLS, and the sequence is presented as
\begin{align} \label{eq_prbs_waveform_vector}
    \vec{p}_{\rm prbs} = [p_1, p_2, \ldots, p_{N_{\rm chip}}]^{\mathsf{T}},
\end{align}
where $p_i \in \{0, 1\}$ is a binary chip value. $T_{\rm chip}$ denotes the chip duration and $T_{\rm prbs} = N_{\rm chip} \times T_{\rm chip}$ denotes the MLS duration, i.e., the duration of the fast-time, and slow-time domain, respectively.

\vspace{-1em}
\subsection{Multi-bit Slot-Partitioned ISAC Waveform Design}
\label{anchor_comment_2_1}
\label{anchor_comment_2_m12}
\label{anchor_comment_E_1}
Denote $\vec{b} = [b_1, b_2, ..., b_{N_{\rm bit}}]^{\mathsf{T}} \in \{0, 1\}^{N_{\rm bit} \times 1}$ as the random bit data sequence, where $N_{\rm bit}$ is the number of bits. The bipolar data symbols are defined as $s_n = 2b_n - 1 \in \{-1, +1\}$. Similar to data symbols, the bipolar mapped PRBS is mapped to $p_i: \{0, 1\} \mapsto \{-1, 1\}$.
While the MLS is traditionally used in radar for its excellent periodic autocorrelation, it often carries only one bit of data in ISAC designs \cite{liu2022integrated, chen2022joint}.
From the MLS in \eqref{eq_prbs_waveform_vector}, this work explores an ISAC sequence where the MLS $\vec{p}_{\rm prbs}$ can carry $N_{\rm bit}^{\rm prbs} \ge 1$ bits from $\vec{b}$. 
For a MLS sequence, we split $N_{\rm chip}$ into $N_{\rm slot} = 2 N_{\rm bit}^{\rm prbs}$ equally-length slots. The number of chips per slot is therefore $N_{\rm chip}^{\rm slot} = N_{\rm chip}/N_{\rm slot}$. The duration of a slot is $T_{\rm slot} = N_{\rm chip}^{\rm slot} \times T_{\rm chip}$. Slots are indexed alternately as pilot slots (odd indices) and data slots (even indices). Denote $d_i \in \{-1, +1\}$ as the chip value $i \in \{1, \ldots, N_{\rm chip}\}$ of the waveform, we propose that
\begin{align} \label{eq_d_isac_i}
    d_i &= p_i \times s_n \ \text{if} \ \left\lceil i / N_{\rm chip}^{\rm slot} \right\rceil = 2n \ \text{(data slot)}, \nonumber \\
    d_i &= p_i \ \text{if} \ \left\lceil i / N_{\rm chip}^{\rm slot} \right\rceil = 2n-1 \ \text{(pilot slot)}.
\end{align}
where $n \in \{1, 2, \dots, N_{\rm bit}^{\rm prbs}\}$ is the data symbol index associated with the current ISAC sequence. This formulation corresponds to a symbol-level spreading structure, where the PRBS sequence acts as a deterministic radar code and the data symbols modulate its sign in the data slots. Let $d(t) = \sum_{i=1}^{N_{\rm chip}} d_i \operatorname{rect} ( (t - (i-1) T_{\rm chip})/T_{\rm chip} )$ be the continuous baseband ISAC signal.
For each time step $t$, the ISAC sequence is given by $\vec{d}_{\mathrm{isac}}(t) \in \{-1+0j, 1+0j\}^{N_{\rm chip} \times 1}$, can be presented as
\begin{align} \label{eq_d_isac_t}
    \vec{d}_{\mathrm{isac}}(t) = [d_1(t), d_2(t), \ldots, d_{\rm N_{\rm chip}}(t)]^{\mathsf{T}},
\end{align}
We transmit the ISAC waveform over $N_{\rm prbs}$ time steps to carrying $N_{\rm bit} = N_{\rm prbs} \times N_{\rm bit}^{\rm prbs}$ (i.e., slow time domain) data symbols to form the transmit ISAC waveform matrix
\begin{align}
    \mathbf{D}_{\rm isac} = [\vec{d}_{\mathrm{isac}}(t_1) , \vec{d}_{\mathrm{isac}}(t_2) , \ldots , \vec{d}_{\mathrm{isac}}(t_{N_{\rm prbs}})]
\end{align}
where $\mathbf{D}_{\rm isac} \in \{-1+0j, +1+0j\}^{N_{\rm chip} \times N_{\rm prbs}}$ is the one frame, time step $t_i = i T_{\rm prbs}$. The duration of the frame is $T_{\rm D} = N_{\rm prbs} \times T_{\rm prbs}$.

\begin{remark}[Multi-bit slot-partitioned ISAC waveform] \label{remark_slot_partitioning}
\label{anchor_comment_2_2}
\label{anchor_comment_E_2}
Conventional MLS-based ISAC waveforms flip the sign of the entire sequence per symbol, $\vec{d}_i = s_i \vec{p}_{\rm prbs}$, carrying one data symbol per code. In contrast, partitioning the $N_{\rm chip}$ chips into $N_{\rm slot} = 2N_{\rm bit}^{\rm prbs}$ alternating pilot/data slots of $N_{\rm chip}^{\rm slot} = N_{\rm chip}/N_{\rm slot}$ chips, as in \eqref{eq_d_isac_i}, spreads $N_{\rm bit}^{\rm prbs}$ symbols independently over the data slots, each with per-bit spreading gain $N_{\rm chip}^{\rm slot}$. Since every chip retains the deterministic code $p_i$, the periodic autocorrelation exploited for range-Doppler sensing is preserved, and the sign flips introduce only data-dependent sidelobes, later suppressed by RDPDNet.
\end{remark}

\vspace{-1em}
\subsection{Sensing and Communication Channel Models}

\subsubsection{Sensing Channel Model}
The ISAC waveform $\vec{d}_{\rm isac}(t)$ in time step $t$ is transmitted from the Tx, impinges on target~$m \in \{1, \ldots, N_{\rm tars }+N_{\rm scats}\}$, and the echo is received at the co-located Rx.
Denote the complex sensing channel value of the target $m$ in time step $t$  as $h_{\rm sen}^{(m)}(t)$, it is expanded to
\begin{align}
    h_{\rm sen}^{(m)}(t)
    = &\sqrt{\sigma_{\rm tars}^{(m)}   \mathrm{PL}_{\rm sen}^{(m)}   P_{\rm tx}   G_{\rm tx}   A_{e}^{\rm rx}} \;
    a_{\rm rx}(\phi^{(m)}_{\rm AOA}, \theta^{(m)}_{\rm ZOA}) \nonumber \\
    & \times a_{\rm tx}(\phi^{(m)}_{\rm AOD}, \theta^{(m)}_{\rm ZOD})
    e^{-j2\pi f_c \tau_{\rm sen}^{(m)}}
    e^{j2\pi f_{{\rm D},{\rm sen}}^{(m)} t},
    \label{eq_H_sen_m}
\end{align}
where $\sigma_{\rm tars}^{(m)} \in \{[\vec{\sigma}_{\rm tars}]_i, [\vec{\sigma}_{\rm scats}]_j\}$, $i = \{1,..,N_{\rm tars}\}, j = \{1,...,N_{\rm scats}\}$, is the radar cross-section (RCS) of the target $m$. 
The round-trip free-space path loss $\text{PL}_{\rm sen}^{(m)} = 1/ \left((4\pi)^2 r_m^4 \right)$, where $r_m = \|\Delta\vec{l}_m\|_2$ is the range to the target $m$ and $\Delta\vec{l}_m = \vec{l}_m - \vec{l}_{\rm tx}$ is the displacement vector from the Tx to the target.
The Rx and Tx scalar phase factor are, respectively, $a_{\rm rx}(\phi^{(m)}_{\rm AOA}, \theta^{(m)}_{\rm ZOA})$ and $a_{\rm tx}(\phi^{(m)}_{\rm AOD}, \theta^{(m)}_{\rm ZOD})$, where $(\phi, \theta)$ denote the azimuth and zenith angles of arrival/departure, respectively.  
$A_{e}^{\rm rx} = \lambda^2 G_{\rm rx} / \left(4\pi\right)$ is the effective aperture value of the Rx array, where $\lambda = c / f_c$ is the carrier wavelength.
$\tau_{\rm sen}^{(m)} = 2 r_m / c$ is the round-trip delay incurred by the target $m$,and $f_{{\rm D},{\rm sen}}^{(m)} = - 2 v_m / \lambda$ is Doppler frequency induced by the radial motion of the target $m$, where $v_m = \vec{v}_m^{\mathsf{T}}\vec{r}_m$ is the radial velocity
of the object projected onto the unit range vector
$\vec{r}_m = \Delta\vec{l}_m / r_m$.
\subsubsection{Communication Channel Model}
\label{anchor_comment_1_2}
\label{anchor_comment_2_m4}
\label{anchor_comment_4_1}
\label{anchor_comment_E_5}
The ISAC waveform $\vec{d}_{\rm isac}(t)$ at time step $t$ reaches the UE via two propagation mechanisms: a direct line-of-sight~(LOS) path from the Tx, and $N^{\rm ref}_{\rm tars} \leq N_{\rm tars} + N_{\rm scats}$ non-line-of-sight~(NLOS) paths.

According to \cite[Eq. 3.6]{Goldsmith2005-ek}, \cite[Eq. 2.27, 2.28]{Tse2005-ei}, denote $h_{\rm com}^{\rm los}(t)$ as the complex LOS channel value from the Tx to the UE at the time step $t$, we have
\begin{align}
	    \label{eq_h_los}
    h_{\rm com}^{\rm los}(t)
    = &\sqrt{\text{PL}_{\rm los} P_{\rm tx} G_{\rm tx} A_{e}^{\rm ue}} 
      a_{\rm ue}(\phi^{\rm ue}_{\rm AOA}, \theta^{\rm ue}_{\rm ZOA})
      \\ \nonumber &\times a_{\rm tx}(\phi^{\rm ue}_{\rm AOD}, \theta^{\rm ue}_{\rm ZOD})
      e^{-j2\pi f_c \tau_{\rm ue}}
      e^{+j2\pi f_{\rm D}^{\rm ue} t}
    ,
\end{align}
where $\text{PL}_{\rm los} = \big(4\pi r_{\rm ue}^2\big)^{-1}$ is the one-way free-space path loss from the Tx to the UE, with range $r_{\rm ue} = \|\Delta\vec{l}_{\rm ue}\|$ and displacement $\Delta\vec{l}_{\rm ue} = \vec{l}_{\rm ue} - \vec{l}_{\rm tx}$. The UE and Tx scalar phase factors is $a_{\rm ue}(\phi^{\rm ue}_{\rm AOA}, \theta^{\rm ue}_{\rm ZOA})$ and $a_{\rm tx}(\phi^{\rm ue}_{\rm AOD}, \theta^{\rm ue}_{\rm ZOD})$, where $(\phi, \theta)$ denote the azimuth and zenith angles of arrival/departure, respectively. $A_{e}^{\rm ue} = \lambda^2 G_{\rm ue} / \left(4\pi\right)$ is the effective aperture value of the UE array.
$\tau_{\rm ue} = r_{\rm ue}/c$ is the one-way propagation delay, and $f_{\rm D}^{\rm ue} = -v_{\rm ue}/\lambda$ is the Doppler shift induced by the UE motion, where $v_{\rm ue} = \vec{v}_{\rm ue}^{\mathsf{T}}\vec{r}_{\rm ue}$ is the radial velocity of the UE projected onto the unit direction vector $\vec{r}_{\rm ue} = \Delta\vec{l}_{\rm ue}/r_{\rm ue}$.

Denote $h_{\rm com}^{{\rm nlos},(m)}(t)$ as the complex NLOS channel value from the Tx through the reflection of the target $m \in \{1, \ldots, N^{\rm ref}_{\rm tars}\}$ to the UE at time step $t$

\begin{align}
    &h_{\rm com}^{{\rm nlos},(m)}(t)
    = \sqrt{\sigma_{\rm tars}^{(m)} \text{PL}_{\rm nlos}^{(m)} P_{\rm tx} G_{\rm tx} A_{e}^{\rm ue}}  a_{\rm ue}(\phi^{(m)}_{\rm AOA}, \theta^{(m)}_{\rm ZOA}) \nonumber \\ 
    &\quad\times a_{\rm tx}(\phi^{(m)}_{\rm AOD}, \theta^{(m)}_{\rm ZOD})
      e^{-j2\pi f_c \tau_{\rm nlos}^{(m)}}
      e^{j2\pi f_{{\rm D},{\rm nlos}}^{(m)} t}
    ,
    \label{eq_h_nlos}
\end{align}
where $\sigma_{\rm tars}^{(m)}$ is the RCS of the target $m$, $\mathrm{PL}_{\rm nlos}^{(m)} = \big( (4\pi)^2 r_{{\rm tx},m}^2 r_{m,{\rm ue}}^2\big)^{-1}$ is the two-hop path loss over the Tx $\to$ target $m$ $\to$ UE. $r_{{\rm tx},m} = \|\vec{l}_m - \vec{l}_{\rm tx}\|$ is the distance from the Tx to the target $m$ and $r_{m,{\rm ue}} = \|\vec{l}_{\rm ue} - \vec{l}_m\|$  is the one from the target $m$ to the UE. The UE and Tx scalar phase factors associated with the reflector $m$ are $a_{\rm ue}(\phi^{(m)}_{\rm AOA}, \theta^{(m)}_{\rm ZOA})$ and $a_{\rm tx}(\phi^{(m)}_{\rm AOD}, \theta^{(m)}_{\rm ZOD})$. 
The two-hop propagation delay is $\tau_{\rm nlos}^{(m)} = (r_{{\rm tx},m} + r_{m,{\rm ue}})/c$, and the associated Doppler shift accumulates the radial velocities of both the reflector and the UE, $f_{{\rm D},{\rm nlos}}^{(m)} = -(v_m + v_{\rm ue})/\lambda$. The $N_{\rm tars}^{\rm ref}$ reflectors are those elements of the scene that contribute a resolvable Tx $\to$ reflector $\to$ UE path, and are drawn from both the $N_{\rm tars}$ moving targets and the $N_{\rm scats}$ stationary scatterers, the latter being the majority in the considered urban geometry.  The $\sigma_{\rm tars}^{(m)}$ denotes the RCS of reflector $m$.

\subsection{Characterization of the Geometry-Based Channel} \label{sec_fading}
\label{anchor_comment_E_5b}
The composite communication channel observed at the UE in time step $t$ aggregates the deterministic LOS path in \eqref{eq_h_los} and the $N_{\rm tars}^{\rm ref}$ NLOS reflections in \eqref{eq_h_nlos}, denoted as $h_{\rm com}(t) = h_{\rm com}^{\rm los}(t) + h_{\rm com}^{\rm nlos}(t)$, where the NLOS component is $h_{\rm com}^{\rm nlos}(t) = \sum_{m=1}^{N_{\rm tars}^{\rm ref}} h_{\rm com}^{{\rm nlos},(m)}(t)$.
Each NLOS path coefficient in \eqref{eq_h_nlos} is determined by the scene geometry through the path loss, RCS, angles, delay, and Doppler, and is thus deterministic given a fixed scene realization and time step. Across independently generated scenes, the reflector locations and the prescribed UE/target mobility vary the geometry-dependent delays, angles, and phases. Since $f_c \tau_{\rm nlos}^{(m)} \gg 1$ at the mmWave carrier $f_c$, small geometry changes produce large carrier-phase changes, so under sufficiently diverse scene geometries the phases $\{2\pi f_c \tau_{\rm nlos}^{(m)} \bmod 2\pi\}_{m=1}^{N_{\rm tars}^{\rm ref}}$ can be approximated across the ensemble as weakly correlated and approximately uniform over $[0,2\pi)$. For a sufficiently large number of contributors, a central-limit approximation then models the aggregate diffuse component as a zero-mean circularly symmetric complex Gaussian, giving its distribution $h_{\rm com}^{\rm nlos}(t) \sim \mathcal{CN}\!\left(0,\, 2\sigma_{\rm d}^2\right)$, where $2\sigma_{\rm d}^2 = \sum_{m=1}^{N_{\rm tars}^{\rm ref}} \sigma_{\rm tars}^{(m)} \, \mathrm{PL}_{\rm nlos}^{(m)} \, P_{\rm tx} \, G_{\rm tx} \, A_{e}^{\rm ue}$,  its absolute value $|h_{\rm com}^{\rm nlos}(t)|$ follows a Rayleigh distribution with probability density function (PDF)
\begin{align} \label{eq_rayleigh_pdf}
    f_{|h_{\rm com}^{\rm nlos}|}(x) = \frac{x}{\sigma_{\rm d}^2}  e^{-x^2 / (2\sigma_{\rm d}^2)}, \quad x \ge 0.
\end{align}
Superimposing the deterministic LOS term $h_{\rm com}^{\rm los}(t)$, with $\nu = |h_{\rm com}^{\rm los}(t)|$, onto the Rayleigh diffuse component yields a Rician-distributed envelope $|h_{\rm com}(t)|$, whose PDF is
\begin{align} \label{eq_rician_pdf}
    f_{|h_{\rm com}|}(x) = \frac{x}{\sigma_{\rm d}^2}  e^{-(x^2 + \nu^2)/(2\sigma_{\rm d}^2)}   I_0 \left(\frac{x \nu}{\sigma_{\rm d}^2}\right), \quad x \ge 0,
\end{align}
where $I_0(\cdot)$ is the zeroth-order modified Bessel function of the first kind. Assuming unit-magnitude scalar array responses, as used in the power-level characterization, the corresponding Rician $K$-factor is the ratio of the deterministic LOS power to the aggregate diffuse power, $K_{\rm ric} = |h_{\rm com}^{\rm los}(t)|^2/(2\sigma_{\rm d}^2) = \mathrm{PL}_{\rm los}   P_{\rm tx}   G_{\rm tx}   A_{e}^{\rm ue} / \big(\sum_{m=1}^{N_{\rm tars}^{\rm ref}} \sigma_{\rm tars}^{(m)}   \mathrm{PL}_{\rm nlos}^{(m)}   P_{\rm tx}   G_{\rm tx}   A_{e}^{\rm ue}\big)$.
This Rician law is a level description across random scenes, not per-path small-scale fading: each path retains its deterministic coefficient in \eqref{eq_h_nlos}, so the BER analysis of Section~\ref{sec_com} averages each geometry-determined realization by MC over scenes rather than assuming a parametric fading distribution.

\subsection{Received ISAC Signal Model}
\subsubsection{Received Communication Signal}
\begin{figure*}[htp!]
\begin{equation}
	\color{black}
\begin{aligned}
    &\vec{y}_{\rm com}(t) = \sqrt{\text{PL}_{\rm los} P_{\rm tx} G_{\rm tx} A_{e}^{\rm ue}} 
      a_{\rm ue}(\phi^{\rm ue}_{\rm AOA}, \theta^{\rm ue}_{\rm ZOA})   a_{\rm tx}(\phi^{\rm ue}_{\rm AOD}, \theta^{\rm ue}_{\rm ZOD})
      e^{-j2\pi f_c \tau_{\rm ue}}
      e^{+j2\pi f_{\rm D}^{\rm ue} t} \vec{d}_{\rm isac}(t - \tau_{\rm ue})   + \vec{n}_{\rm com}(t)
       \\ &+ \textstyle\sum_{m = 1}^{N_{\rm tars}^{\rm ref}} \sqrt{\sigma_{\rm tars}^{(m)} \text{PL}_{\rm nlos}^{(m)} P_{\rm tx} G_{\rm tx} A_{e}^{\rm ue}}  a_{\rm ue}(\phi^{(m)}_{\rm AOA}, \theta^{(m)}_{\rm ZOA})
        a_{\rm tx}(\phi^{(m)}_{\rm AOD}, \theta^{(m)}_{\rm ZOD})
      e^{-j2\pi f_c \tau_{\rm nlos}^{(m)}}
      e^{+j2\pi f_{{\rm D},{\rm nlos}}^{(m)} t} \vec{d}_{\rm isac}(t - \tau_{\rm nlos}^{(m)})
\end{aligned}
\label{eq_vec_y_com}
\end{equation}
\hrule
\begin{equation}
	\color{black}
\begin{aligned}
    \vec{y}_{\rm sen}(t)  &=  \textstyle\sum_{m = 1}^{N_{\rm tars} + N_{\rm scats}}
         \sqrt{\sigma_{\rm tars}^{(m)}   \mathrm{PL}_{\rm sen}^{(m)}   P_{\rm tx}   G_{\rm tx}   A_{e}^{\rm rx}} \;
    a_{\rm rx}(\phi^{(m)}_{\rm AOA}, \theta^{(m)}_{\rm ZOA}) a_{\rm tx}(\phi^{(m)}_{\rm AOD}, \theta^{(m)}_{\rm ZOD})  \\
    &\quad\times 
    e^{-j2\pi f_c \tau_{\rm sen}^{(m)}}
    e^{j2\pi f_{{\rm D},{\rm sen}}^{(m)} t} \vec{d}_{\rm isac}( t - \tau_{\rm sen}^{(m)})
       + \vec{n}_{\rm sen}(t)
\end{aligned}
\label{eq_vec_y_sen}
\end{equation}
\hrule
\end{figure*}
At each time step $t$, the transmitted signal contains $N_{\rm chip}$ chips. Following \cite[Eq. 2.17]{Tse2005-ei} and \cite[Eq. 3.5]{Goldsmith2005-ek}, the complex signal received at the UE, $\vec{y}_{\rm com}(t) \in \mathbb{C}^{N_{\rm chip} \times 1}$, consists of the direct LOS contribution, the superposition of $N_{\rm tars}^{\rm ref}$ NLOS paths, and the additive Gaussian noise vector $\vec{n}_{\rm com}(t) \sim \mathcal{CN}(\vec{0}, \sigma^2_{\rm com}\mathbf{I}_{N_{\rm chip}})$ with covariance matrix $\sigma^2_{\rm com}\mathbf{I}_{N_{\rm chip}}$, as shown in \eqref{eq_vec_y_com}. We discretize each time step into $N_{\rm chip}$ chips to form the UE response vector $\vec{y}_{{\rm com},i} \in \mathbb{C}^{N_{\rm chip} \times 1}$, given by
\begin{align}
    \vec{y}_{{\rm com},i}
     =  \vec{d}_{{\rm isac},i}^{\rm los}h_{{\rm com},i}^{\rm los}
       +  \textstyle\sum_{\rm m = 1}^{N_{\rm tars}^{\rm ref}}\vec{d}_{{\rm isac},i}^{\mathrm{nlos}, (m)}h_{{\rm com},i}^{\mathrm{nlos}, (m)}
       +  \vec{n}_{{\rm com},i}
    ,
    \label{eq_vec_y_com_disc}
\end{align}
where $\vec{n}_{{\rm com},i} \in \mathbb{C}^{N_{\rm chip} \times 1}$, $h_{{\rm com},i}^{\rm los}$ is the complex LOS channel coefficient at step $i$, and $h_{{\rm com},i}^{\mathrm{nlos}, (m)}$ is the complex NLOS channel coefficient associated with reflector $m$ at step $i$. The vector $\vec{d}_{{\rm isac},i}^{\rm los} = \vec{d}_{\rm isac}(t_i - \tau_{\rm ue})\in \mathbb{C}^{N_{\rm chip} \times 1}$ is the delayed waveform capturing the direct-path contribution at time step $t_i$.
The corresponding delayed waveform replicas are collected in $\vec{d}_{{\rm  isac},i}^{\mathrm{nlos}, (m)} = \vec{d}_{\rm isac}(t_i - \tau_{\rm nlos}^{(m)})\in \mathbb{C}^{N_{\rm chip} \times 1}$.
Stacking the response vectors over $N_{\rm prbs}$ steps gives the communication response matrix $\mathbf{Y}_{\rm com}\in \mathbb{C}^{N_{\rm chip} \times N_{\rm prbs}}$, which is presented as
\begin{align}
    \mathbf{Y}_{\rm com} 
    = \mathbf{D}_{\rm isac}^{\rm los}\mathbf{H}_{\rm com}^{\rm los}
     +  \textstyle\sum_{m =1}^{N_{\rm tars}^{\rm ref}} \mathbf{D}_{\rm isac}^{{\rm nlos},(m)}\mathbf{H}_{\rm com}^{{\rm nlos},(m)}
     +  \mathbf{N}_{\rm com},
    \label{eq_Y_com}
\end{align}
where the waveform matrices $\mathbf{D}_{\rm isac}^{\rm los} = [\vec{d}_{{\rm isac},1}^{\rm los}, \ldots, \vec{d}_{{\rm isac},N_{\rm prbs}}^{\rm los}]$, $\mathbf{D}_{\rm isac}^{{\rm nlos},(m)} = [\vec{d}_{{\rm isac},1}^{{\rm nlos},(m)}, \ldots, \vec{d}_{{\rm isac},N_{\rm prbs}}^{{\rm nlos},(m)}]$, and the noise matrix $\mathbf{N}_{\rm com} = [\vec{n}_{{\rm com},1}, \ldots, \vec{n}_{{\rm com},N_{\rm prbs}}]$ column-stack the per-step delayed waveforms and noise over the $N_{\rm prbs}$ steps, the last with i.i.d. entries $[\mathbf{N}_{\rm com}]_{k,i} \sim \mathcal{CN}(0, \sigma^2_{\rm com})$, and the diagonal channel matrices $\mathbf{H}_{\rm com}^{\rm los} = \text{diag}\big(\vec{h}_{\rm com}^{\rm los}\big)$ and $\mathbf{H}_{\rm com}^{{\rm nlos},(m)} = \text{diag}\big(\vec{h}_{\rm com}^{{\rm nlos},(m)}\big)$ collect the coefficients $\vec{h}_{\rm com}^{\rm los} = [h_{{\rm com},1}^{\rm los}, \ldots, h_{{\rm com},N_{\rm prbs}}^{\rm los}]^{\mathsf{T}}$ and $\vec{h}_{\rm com}^{{\rm nlos},(m)} = [h_{{\rm com},1}^{{\rm nlos},(m)}, \ldots, h_{{\rm com},N_{\rm prbs}}^{{\rm nlos},(m)}]^{\mathsf{T}}$. At time step $t_{N_{\rm prbs}}$, part of the final transmitted block may arrive during $t_{N_{\rm prbs}+1}$. The transmitter therefore appends the zero block $\vec{d}_{\rm isac}(t_{N_{\rm prbs}+1}) = \vec{0}^{N_{\rm chip} \times 1}$, and the received matrix is extended to $\mathbf{Y}_{\rm com}\in \mathbb{C}^{N_{\rm chip} \times (N_{\rm prbs}+1)}$ to retain the delayed tail of $\vec{s}_{\rm bpsk}$. This one-block extension captures the complete tail provided every propagation delay remains within one block duration, i.e., $\tau_{\rm ue}<T_{\rm prbs}$ and $\max_m\tau_{\rm nlos}^{(m)}<T_{\rm prbs}$.

\subsubsection{Received Sensing Signal}
\label{anchor_comment_4_3}

Following \cite[Eq. 2.2]{Richards2022}, the ISAC signal received at the Rx, $\vec{y}_{\rm sen}(t) \in \mathbb{C}^{N_{\rm chip} \times 1}$, is the superposition of the echoes from all targets and additive noise, as shown in \eqref{eq_vec_y_sen}. The additive Gaussian noise vector is $\vec{n}_{\rm sen}(t) \sim \mathcal{CN}(\vec{0}, \sigma^2_{\rm sen}\mathbf{I}_{N_{\rm chip}})$, with covariance matrix $\sigma^2_{\rm sen}\mathbf{I}_{N_{\rm chip}}$. According to \cite[Eq. 2.26]{Tse2005-ei} and \cite[Eq. 6]{gaudio2020effectiveness}, we discretize each time step into $N_{\rm chip}$ chips to form the Rx response vector $\vec{y}_{{\rm sen},i} \in \mathbb{C}^{N_{\rm chip} \times 1}$,
\begin{align}
    \vec{y}_{{\rm sen},i}
    = \textstyle\sum_{m = 1}^{N_{\rm tars} + N_{\rm scats}}\vec{d}_{{\rm isac},i}^{(m)}   h_{{\rm sen},i}^{(m)}
      + \vec{n}_{{\rm sen},i},
    \label{eq_vec_y_sen_disc}
\end{align}
where $h_{{\rm sen},i}^{(m)}$ is the complex sensing channel of the target $m$ at step $i$,
$\vec{d}_{{\rm isac},i}^{(m)} = \vec{d}_{{\rm isac}}(t_i - \tau_{\rm sen}^{(m)}) \in \mathbb{C}^{N_{\rm chip} \times 1}$ is the sensing delay waveform vector of target $m$ at time step $t_i$,
$\vec{n}_{{\rm sen},i} \in \mathbb{C}^{N_{\rm chip} \times 1}$.
Stacking all $N_{\rm prbs}$ response vectors gives the sensing response matrix $\mathbf{Y}_{\rm sen}\in \mathbb{C}^{N_{\rm chip} \times N_{\rm prbs}}$ at the Rx as
\begin{align}
    \mathbf{Y}_{\rm sen} = \textstyle\sum_{m = 1}^{N_{\rm tars} + N_{\rm scats}} \mathbf{D}_{\rm isac}^{{\rm sen},(m)}\mathbf{H}_{\rm sen}^{(m)}  + \mathbf{N}_{\rm sen},
    \label{eq_Y_sen}
\end{align}
where $\mathbf{D}_{\rm isac}^{{\rm sen},(m)} = [\vec{d}_{{\rm isac},1}^{(m)}, \ldots, \vec{d}_{{\rm isac},N_{\rm prbs}}^{(m)}]$ and $\mathbf{N}_{\rm sen} = [\vec{n}_{{\rm sen},1}, \ldots, \vec{n}_{{\rm sen},N_{\rm prbs}}]$ column-stack the per-step delayed waveforms and noise over the $N_{\rm prbs}$ steps, the latter with i.i.d. entries $[\mathbf{N}_{\rm sen}]_{k,i} \sim \mathcal{CN}(0, \sigma^2_{\rm sen})$, and $\mathbf{H}_{\rm sen}^{(m)} = \text{diag}\big(\vec{h}_{\rm sen}^{(m)}\big)$ with $\vec{h}_{\rm sen}^{(m)} = [h_{{\rm sen},1}^{(m)}, \ldots, h_{{\rm sen},N_{\rm prbs}}^{(m)}]^{\mathsf{T}}$. Throughout, the transmit signal-to-noise ratio is the dimensionless power ratio $\mathrm{SNR}_{\rm Tx} = P_{\rm tx}/\sigma_{\rm sen}^2$ for sensing (and $P_{\rm tx}/\sigma_{\rm com}^2$ for communication), where $\sigma_{\rm sen}^2$ and $\sigma_{\rm com}^2$ are the noise variances (powers).\label{anchor_comment_2_m7}

\section{Conventional Range-Doppler Estimation}

\subsection{Matched Filtering for Range Estimation}

Because the Rx has full knowledge of the transmitted ISAC waveform $\mathbf{D}_{\rm isac}$, a matched filter (MF) correlates the received sensing matrix $\mathbf{Y}_{\rm sen}$ with it directly: the embedded communication data need not be removed, as the data-modulated waveform itself serves as the correlation reference and preserves the autocorrelation peak used for delay estimation. Denote the MF output by $\mathbf{Z}_{\rm sen} \in \mathbb{C}^{N_{\rm chip} \times N_{\rm prbs}}$. Consistent with the fast Fourier transform (FFT)-based implementation, $\circledast$ denotes length-$N_{\rm chip}$ circular correlation. For each step $i$,
\begin{align} \label{eq_Z_sen_i}
    &[\mathbf{Z}_{\rm sen}]_{:,i}
    = [\mathbf{Y}_{\rm sen}]_{:,i}\circledast[\mathbf{D}_{\rm isac}^{*}]_{:,i} \nonumber\\
    &= \textstyle\sum_{m=1}^{N_{\rm tars}+N_{\rm scats}}
    \bigl(\vec d_{{\rm isac},i}^{(m)}\circledast\vec d_{{\rm isac},i}^{*}\bigr)
    h_{{\rm sen},i}^{(m)}+\vec n_{{\rm sen},i}^{\rm mf}.
\end{align}
Here, $\vec n_{{\rm sen},i}^{\rm mf}=\vec n_{{\rm sen},i}\circledast\vec d_{{\rm isac},i}^{*}\sim\mathcal{CN}(\vec 0,\mathbf\Sigma_{{\rm sen},i}^{\rm mf})$ is the filtered noise vector, whose covariance generally satisfies $\mathbf\Sigma_{{\rm sen},i}^{\rm mf}\ne\sigma_{\rm sen}^{2}\mathbf I_{N_{\rm chip}}$. Let $\langle\cdot\rangle_{N_{\rm chip}}$ denote the modulo-$N_{\rm chip}$ index. Define the sensing autocorrelation function (ACF) vector $\vec{\mathrm{ACF}}_{\rm sen}(t_i,\tau_{\rm sen}^{(m)})=\vec d_{{\rm isac},i}^{(m)}\circledast\vec d_{{\rm isac},i}^{*}$, whose $k$-th entry, $k\in\{1,\ldots,N_{\rm chip}\}$, is presented as
\begin{align} \label{eq_acf_i_k}
[\vec{\mathrm{ACF}}_{\rm sen}(t_i,\tau_{\rm sen}^{(m)})]_k = \textstyle \sum_{n=1}^{N_{\rm chip}} \!\! d_n(t_i-\tau_{\rm sen}^{(m)})
d_{\langle n-k\rangle_{N_{\rm chip}}+1}^{*}(t_i).
\end{align}
The estimated target range is $\hat r_m=(c/2)\hat\tau_{\rm sen}^{(m)}$, where $\hat\tau_{\rm sen}^{(m)}=(k-1)T_{\rm chip}$ and $k=\arg\max_k|[\vec{\mathrm{ACF}}_{\rm sen}(t_i,\tau_{\rm sen}^{(m)})]_k|$.

\subsection{Range-Doppler Map Generation and CA-CFAR Detection}
Denote $\mathbf{Z}_{\rm rd} \in \mathbb{C}^{N_{\rm chip} \times N_{\rm prbs}}$ as the RD matrix, where $N_{\rm chip}$ is the range axis and $N_{\rm prbs}$ is the Doppler axis. From the MF output \eqref{eq_Z_sen_i}, we apply the discrete Fourier transform (DFT) for each chip index $k$ over $N_{\rm prbs}$ slow-time steps before shifting to form $[\mathbf{Z}_{\rm rd}]_{k,:} = \mathcal{F}_{\rm shift} \{[\mathbf{Z}_{\rm sen}]_{k,:}\}  \in \mathbb{C}^{1 \times N_{\rm prbs}}$, which is presented as
\begin{align} \label{eq_Z_rd_k_i}
    [\mathbf{Z}_{\rm rd}]_{k,i}  =  & \textstyle \sum_{m=1}^{N_{\rm tars}+N_{\rm scats}} \sum_{n=1}^{N_{\rm prbs}} [\vec{\text{ACF}}_{\mathrm{sen}}(t_n, \tau_{\rm sen}^{(m)})]_k   h_{{\rm sen},n}^{(m)} \nonumber \\ & \textstyle\times \exp(\frac{-j2\pi(i-N_{\rm prbs}/2-1)(n-1)}{N_{\rm prbs}}) + [\mathbf{N}_{{\rm rd}}]_{k,i},
\end{align}
where $k \in \{1,\ldots,N_{\rm chip}\}$ is the range bin index, $i - N_{\rm prbs}/2 - 1 \in \{-N_{\rm prbs}/2, \ldots, N_{\rm prbs}/2 - 1\}$ is the velocity index after shift with zero-Doppler centered at $i = N_{\rm prbs}/2 + 1$, and $[\mathbf{N}_{{\rm rd}}]_{k,i} = \sum_{n=1}^{N_{\rm prbs}} [\vec{n}^{\rm mf}_{{\rm sen},n}]_k \exp(-j2\pi(i-N_{\rm prbs}/2-1)(n-1)/N_{\rm prbs})$ is the noise term after DFT. Since stationary scatterers have zero velocity, i.e., $\vec{v}_{\rm scats} = \vec{0}$, their contribution falls entirely at the center column $i = N_{\rm prbs}/2 + 1$ of $\mathbf{Z}_{\rm rd}$. Therefore, zeroing out this column, i.e., $[\mathbf{Z}_{\rm rd}]_{:,  N_{\rm prbs}/2+1} = \vec{0}^{N_{\rm chip} \times 1}$, effectively suppresses stationary clutter and isolates the moving targets.

Accordingly, since the channel model adopts the Doppler convention $f_{\rm D} = -2v/\lambda$, the physical radial velocity associated with the shifted-DFT column $i$ is $v = -(i - N_{\rm prbs}/2 - 1)\lambda/(2 N_{\rm prbs} T_{\rm prbs})$, rather than $i \Delta v$ on the raw column index.

\begin{remark}
    The cell $[\mathbf{Z}_{\rm rd}]_{k,i}$ achieves its peak at $(k, i)$, whereas $k = \arg\max_{k} |[\vec{\text{ACF}}_{\mathrm{sen}}(t_n, \tau_{\rm sen}^{(m)})]_k|$ in \eqref{eq_acf_i_k}, and $i = \arg\max_{i} \big|\sum_{n=1}^{N_{\rm prbs}} h_{{\rm sen},n}^{(m)}   \exp(-j2\pi (i-N_{\rm prbs}/2-1)(n-1)/N_{\rm prbs})\big|$ corresponds to the range bin and velocity bin of target $m$, respectively.
\end{remark}
Because $[\mathbf{N}_{{\rm rd}}]_{k,i}$ is Gaussian, target peaks in $\mathbf{Z}_{\rm rd}$ can be detected using cell-averaging constant-false-alarm-rate (CA-CFAR) processing \cite[Sec~11.3]{Barkat2005-lz}. Define $\mathbf{P}_{\rm rd}=|\mathbf{Z}_{\rm rd}|^2\in\mathbb{R}^{N_{\rm chip}\times N_{\rm prbs}}$ as the RD power matrix. The implementation uses guard half-widths $N_{\rm g,r}$ and $N_{\rm g,d}$ and training widths $N_{\rm t,r}$ and $N_{\rm t,d}$ along the range and Doppler dimensions, respectively. For each cell under test $(k,i)$, the reference-cell average is $P_{\rm rd}^{(k,i)} = (1/N_{\rm w}) \sum_{(k', i') \in \mathcal{R}_{k,i}} [\mathbf{P}_{\mathrm{rd}}]_{k',i'}$, where $\mathcal{R}_{k,i}$ contains cells satisfying $|k'-k|\leq N_{\rm g,r}+N_{\rm t,r}$ and $|i'-i|\leq N_{\rm g,d}+N_{\rm t,d}$, excluding cells satisfying both $|k'-k|\leq N_{\rm g,r}$ and $|i'-i|\leq N_{\rm g,d}$. For an interior cell, the number of reference cells is
$
N_{\rm w}={}[2(N_{\rm g,r}+N_{\rm t,r})+1][2(N_{\rm g,d}+N_{\rm t,d})+1]
-(2N_{\rm g,r}+1)(2N_{\rm g,d}+1).
$
At a map boundary, the window is clipped and the implementation uses the actual number $N_{\rm w}^{(k,i)}$ of available reference cells. A potential target is declared when $[\mathbf{P}_{\rm rd}]_{k,i}>[\mathbf{\lambda}_{\rm rd}]_{k,i}$, where $[\mathbf{\lambda}_{\rm rd}]_{k,i}=\alpha_{\rm rd}^{(k,i)}P_{\rm rd}^{(k,i)}$ and $\alpha_{\rm rd}^{(k,i)}=N_{\rm w}^{(k,i)}(P_{\rm fa}^{-1/N_{\rm w}^{(k,i)}}-1)$ \cite[Eq.~11.18]{Barkat2005-lz}. To avoid multiple detections from a single target, a declared cell is retained only if it is a local maximum, i.e., its power is not exceeded by any cell within the surrounding $[2(N_{\rm g,r}+N_{\rm t,r})+1]\times[2(N_{\rm g,d}+N_{\rm t,d})+1]$ window, yielding one detection per target.

\begin{remark}
\label{anchor_comment_2_m6}
\label{anchor_comment_2_m10}
\label{anchor_comment_4_2}
\label{anchor_comment_4_6}
    Given the RD matrix $\mathbf{Z}_{\rm rd}$, the range resolution
    along the range axis (indexed by $k$) is $\Delta r = c T_{\rm chip}/2$,
    and the velocity resolution along the Doppler axis (indexed by $i$) is
    $\Delta v = \lambda/(2 N_{\rm prbs} T_{\rm prbs})$.
\end{remark}

\subsection{Sinc Interpolation for Sub-Bin Refinement} 
\label{subsec_sinc_inter}
The estimated range $r_m$ and velocity $v_m$ of target $m$ on the RD heatmap are limited by the range and velocity bin resolutions $\Delta r$ and $\Delta v$, respectively. Hence, the integer-bin peak location introduces quantization error. To reduce this error, we apply sinc interpolation to reconstruct the local complex RD response on a fractional-bin grid.
Let $(k^{(m)},i^{(m)})$ denote the detected integer range-Doppler peak of target $m$ in $\mathbf{Z}_{\rm rd}$, where $k^{(m)}$ and $i^{(m)}$ are the range-bin and Doppler-bin indices, respectively. Around this peak, define
$
\widetilde{Z}_{r}^{(m)}(u)
= \sum_{k=k^{(m)}-N_{\rm sinc}}^{k^{(m)}+N_{\rm sinc}}
[\mathbf{Z}_{\rm rd}]_{k,i^{(m)}}\operatorname{sinc}(u-k),
\widetilde{Z}_{v}^{(m)}(u) 
= \sum_{i=i^{(m)}-N_{\rm sinc}}^{i^{(m)}+N_{\rm sinc}}
[\mathbf{Z}_{\rm rd}]_{k^{(m)},i}\operatorname{sinc}(u-i),
$
where $N_{\rm sinc}$ is the half-width of the interpolation window in bins, and $u$ denotes a fractional bin index. The refined fractional-bin locations are obtained by maximizing the reconstructed magnitude in the local half-bin interval as
$
u_r^{(m)}
=
\arg\max_{u\in[k^{(m)}-\frac{1}{2}, k^{(m)}+\frac{1}{2}]}
|\widetilde{Z}_{r}^{(m)}(u)|,
u_v^{(m)}
=
\arg\max_{u\in[i^{(m)}-\frac{1}{2}, i^{(m)}+\frac{1}{2}]}
|\widetilde{Z}_{v}^{(m)}(u)|.
$
Then, using the one-based bin indexing convention, the refined range estimate is $\hat r_m = (u_r^{(m)}-1)\Delta r$. For the Doppler dimension, the RD map is FFT-shifted, so the centered Doppler index must be used. For even $N_{\rm prbs}$, the refined velocity estimate is $\hat v_m=-\lambda(u_v^{(m)}-\frac{N_{\rm prbs}}{2}-1)/(2N_{\rm prbs}T_{\rm prbs})$.

\section{Proposed Probabilistic Denoising-Enhanced Sensing Framework}
\label{anchor_comment_1_1}

\subsection{Denoising Problem Formulation}

The data-bearing slots make the sensing ACF fluctuate across slow time and thereby introduce spurious peaks in the RD map. Define the clean PRBS correlation $\vec{\mathrm{ACF}}_{\rm prbs}$ and the data-induced perturbation
$
\vec{\mathrm{ACF}}_{\rm ran}=\vec{\mathrm{ACF}}_{\rm sen}-\vec{\mathrm{ACF}}_{\rm prbs}.
$
Applying to \eqref{eq_acf_i_k}, the MF \eqref{eq_Z_sen_i} and RD-domain outputs \eqref{eq_Z_rd_k_i} can then be decomposed as
$
\mathbf Z_{\rm sen}=\mathbf Z_{\rm prbs}+\mathbf Z_{\rm ran}+\mathbf N_{\rm sen}^{\rm mf},
\mathbf Z_{\rm rd}=\mathbf Z_{\rm rd}^{\rm prbs}+\mathbf Z_{\rm rd}^{\rm ran}+\mathbf N_{\rm rd}^{\rm mf}.
$
Here $\mathbf Z_{\rm rd}^{\rm prbs}$ is the desired clean response, while $\mathbf Z_{\rm rd}^{\rm ran}+\mathbf N_{\rm rd}^{\rm mf}$ is the structured data-induced and thermal disturbance. We seek an estimate that remains faithful to the observation while satisfying a clean-map prior $\mathcal R(\cdot)$:
\begin{subequations} \label{eq_problem}
\begin{alignat}{2}
&&& \mathcal P_1:\ \min_{\hat{\mathbf Z}_{\rm rd}}\
\|\mathbf Z_{\rm rd}-\hat{\mathbf Z}_{\rm rd}\|_F^2
+\lambda_{\mathcal R}\mathcal R(\hat{\mathbf Z}_{\rm rd}) \label{eq_problem_obj} \\
& \text{s.~t.~}
&& \hat{\mathbf Z}_{\rm rd}\in\mathbb C^{N_{\rm chip}\times N_{\rm prbs}}, \label{eq_problem_domain} \\
&&& \mathbf M_{\rm tars}\odot\hat{\mathbf Z}_{\rm rd}
=\mathbf M_{\rm tars}\odot\mathbf Z_{\rm rd}^{\rm prbs}, \label{eq_problem_support} \\
&&& \lambda_{\mathcal R}\ge 0. \label{eq_problem_lambda}
\end{alignat}
\end{subequations}
where the target-support mask $\mathbf M_{\rm tars}\in[0,1]^{N_{\rm chip}\times N_{\rm prbs}}$, derived from the clean map $\mathbf Z_{\rm rd}^{\rm prbs}$, marks the true target locations, so the second constraint enforces fidelity on the target support while the prior $\mathcal R(\cdot)$ suppresses the data-induced and thermal disturbance elsewhere. RDPDNet learns $\mathcal R(\cdot)$ implicitly from paired corrupted and clean RD maps.

\subsection{Probabilistic Denoising Network} \label{sec_supervised_learning}
\label{anchor_comment_3_4}
\label{anchor_comment_E_8}

To solve the $\mathcal{P}_1$, we need to know the knowledge of 
$\epsilon_{\rm sen}$ with the unknown target channels and noise 
variance. In this work, we propose a denoising pipeline that directly minimizes the residual against clean RD maps $\mathbf{Z}_{\rm rd}^{\rm prbs}$, implicitly learning the noise level and the slot-structured disturbance covariance without any explicit knowledge of $\epsilon_{\rm sen}$, $\sigma_{\rm sen}^2$, or target locations. From \eqref{eq_problem}, we can reformulate the denosing problem to
\begin{subequations} \label{eq_problem_2}
\begin{alignat}{2}
&&& \mathcal{P}_2:\ \min_{\mathbf{\Phi}_{\rm rd}} \;
\mathbb E \left[\|\mathbf{Z}_{\rm rd}^{\rm prbs} - \mathbf{\hat{Z}}_{\rm rd}\|_F^2\right] \label{eq_problem_2_obj} \\
& \text{s.~t.~}
&& \mathbf{\hat{Z}}_{\rm rd} = f_{\rm den}(\mathbf{Z}_{\rm rd}; \mathbf{\Phi}_{\rm rd}) \in \mathbb{C}^{N_{\rm chip} \times N_{\rm prbs}}, \label{eq_problem_2_realize} \\
&&& q(\vec{\mathfrak{z}}_l|\vec{\mathfrak{z}}_{l-1}) = \mathcal{N}(\boldsymbol{\mu}_{\vec{\mathfrak{z}}_{l-1}}, \mathbf{\Sigma}_{\vec{\mathfrak{z}}_{l-1}}), \; l = 1, \ldots, L, \label{eq_problem_2_gauss} \\
&&& D_{\text{KL}}\big(q(\vec{\mathfrak{z}}_l|\vec{\mathfrak{z}}_{l-1})   \|   p(\vec{\mathfrak{z}}_l|\vec{\mathfrak{z}}_{l+1})\big) \le \epsilon_{\rm KL}. \label{eq_problem_2_kl}
\end{alignat}
\end{subequations}
where the hierarchical latent distributions $q(\cdot)$, $p(\cdot)$ and the number of latent levels $L$ are defined below, and $\epsilon_{\rm KL} \ge 0$ bounds the latent divergence.
The denoiser is $\hat{\mathbf Z}_{\rm rd}=f_{\rm den}(\mathbf Z_{\rm rd};\mathbf\Phi_{\rm rd})$. Define $\mathcal T(\mathbf Z)=[\Re\{\mathbf Z\},\Im\{\mathbf Z\}]$ as its two-channel real representation. Thus, the network tensors belong to $\mathbb R^{2\times N_{\rm chip}\times N_{\rm prbs}}$; the equations below retain the original RD-map symbols for readability.

In the RDPDNet architecture design, we develop extractor, encoder, and decoder modules that are iteratively used to learn the latent distributions \(\{\vec{\mathfrak{z}}_1, \ldots, \vec{\mathfrak{z}}_L\}\), where $L$ is the number of latent vectors. The encoder modules extract the latent vectors from $\mathbf{Z}_{\rm rd}^{\rm prbs}$. Defined \(f_{\rm g, enc}(\mathbf{Z}_{\rm rd}^{\rm prbs}; \mathbf{\Phi}_{\rm enc}) \sim q(\vec{\mathfrak{z}}_1, \ldots, \vec{\mathfrak{z}}_L|\mathbf{Z}_{\rm rd}^{\rm prbs})\) as the hierarchical encoder function, where $\mathbf{\Phi}_{\rm enc}$ is the updated parameter of the encoder layer. The approximate likelihood probability of the hierarchical process can be presented as 
\begin{align}
    \textstyle q(\vec{\mathfrak{z}}_1, \ldots, \vec{\mathfrak{z}}_L|\mathbf{Z}_{\rm rd}^{\rm prbs}) =  q(\vec{\mathfrak{z}}_1|\mathbf{Z}_{\rm rd}^{\rm prbs}) \prod_{l=2}^{L} &q(\vec{\mathfrak{z}}_l| \vec{\mathfrak{z}}_{l-1}) ,
\end{align}
where \(q(\vec{\mathfrak{z}}_1|\mathbf{Z}_{\rm rd}^{\rm prbs}) = \mathcal{N}(\vec{\mathfrak{z}}_1;  \boldsymbol{\mu}_{\mathbf{Z}_{\rm rd}^{\rm prbs}}, \mathbf{\Sigma}_{\mathbf{Z}_{\rm rd}^{\rm prbs}})\), \(q(\vec{\mathfrak{z}}_l|\vec{\mathfrak{z}}_{l-1}) = \mathcal{N}(\vec{\mathfrak{z}}_l;  \boldsymbol{\mu}_{\vec{\mathfrak{z}}_{l-1}}, \mathbf{\Sigma}_{\vec{\mathfrak{z}}_{l-1}})\). Given the observation RD matrix input $\mathbf{Z}_{\rm rd}$, we define the extraction function $f_{\rm g, ext}(\mathbf{Z}_{\rm rd}; \mathbf{\Phi}_{\rm ext}) \sim p(\vec{\mathfrak{z}}_L|\mathbf{Z}_{\rm rd})$, updated from parameter $\mathbf{\Phi}_{\rm ext}$, where $p(\vec{\mathfrak{z}}_L|\mathbf{Z}_{\rm rd}) = \mathcal{N}(\boldsymbol{\mu}_{\mathbf{Z}_{\rm rd}}, \mathbf{\Sigma}_{\mathbf{Z}_{\rm rd}})$. From the estimated latent vectors from the encoder and extractor, we apply the decoder function $\mathbf{\hat{Z}}_{\rm rd} = f_{\rm g, dec}(\vec{\mathfrak{z}}_1, \vec{\mathfrak{z}}_2, \ldots, \vec{\mathfrak{z}}_L; \mathbf{\Phi}_{\rm dec}) \sim p(\mathbf{\hat{Z}}_{\rm rd}, \vec{\mathfrak{z}}_1, \vec{\mathfrak{z}}_2, \dots, \vec{\mathfrak{z}}_L)$ to the estimated RD heatmap $\mathbf{\hat{Z}}_{\rm rd}$. The distribution is expanded to
\begin{align}
    p(\mathbf{\hat{Z}}_{\rm rd}, \vec{\mathfrak{z}}_1, \dots, \vec{\mathfrak{z}}_L)  =  p(\mathbf{\hat{Z}}_{\rm rd}|\vec{\mathfrak{z}}_1)   \textstyle \prod_{l=1}^{L-1} p(\vec{\mathfrak{z}}_l|\vec{\mathfrak{\mathfrak{z}}}_{l + 1}),
\end{align}
where \(p(\mathbf{\hat{Z}}_{\rm rd}|\vec{\mathfrak{z}}_1)\) is the decoder output, modeled as a latent vector \(\vec{\mathfrak{z}}_1 \sim \mathcal{N}(\vec{\mathfrak{z}}_1;  \boldsymbol{\mu}_{\vec{\mathfrak{z}}_2}, \mathbf{\Sigma}_{\vec{\mathfrak{z}}_2})\), \(p(\vec{\mathfrak{z}}_l|\vec{\mathfrak{z}}_{l+1}) = \mathcal{N}(\vec{\mathfrak{z}}_l;  \boldsymbol{\mu}_{\vec{\mathfrak{z}}_{l+1}}, \mathbf{\Sigma}_{\vec{\mathfrak{z}}_{l+1}})\) is the conditional prior for the levels \(l \in \{1, \ldots, L-1\}\). The parameter update during the training is $\mathbf{\Phi}_{\rm rd} = \{\mathbf{\Phi}_{\rm enc}, \mathbf{\Phi}_{\rm ext}, \mathbf{\Phi}_{\rm dec}\}$, while the parameter is used for inference is $\mathbf{\Phi}_{\rm rd} = \{\mathbf{\Phi}_{\rm ext}, \mathbf{\Phi}_{\rm dec}\}$, which can reduce the complexity of the RDPDNet model. Splitting the denoise function $f_{\rm den}(.)$ to train and inference stage, they are
\begin{align} 
    \mathbf{\hat{Z}}_{\rm rd} =&\ f_{\rm den}^{\rm train}(\mathbf{Z}_{\rm rd}, \mathbf{Z}_{\rm rd}^{\rm prbs}; \mathbf{\Phi}_{\rm rd}) \label{eq_train_denoising_process} \\
    =&\ f_{\rm g, dec} \Big(f_{\rm g, ext}(\mathbf{Z}_{\rm rd}; \mathbf{\Phi}_{\rm ext}), f_{\rm g, enc}(\mathbf{Z}_{\rm rd}^{\rm prbs}; \mathbf{\Phi}_{\rm enc}); \mathbf{\Phi}_{\rm dec} \Big), \nonumber \\
    \mathbf{\hat{Z}}_{\rm rd} =&\ f_{\rm den}^{\rm infer}(\mathbf{Z}_{\rm rd}; \mathbf{\Phi}_{\rm rd}) \nonumber \\
    =&\ f_{\rm g, dec} \Big(f_{\rm g, ext}(\mathbf{Z}_{\rm rd}; \mathbf{\Phi}_{\rm ext}); \mathbf{\Phi}_{\rm dec} \Big). \label{eq_infer_denoising_process}
\end{align}

\subsection{Adversarial Frequency Mixup and Loss Design}
\label{anchor_comment_1_5}
\label{anchor_comment_2_7}
\label{anchor_comment_2_m8}
\label{anchor_comment_3_5}
\label{anchor_comment_E_7}
\begin{figure*}[!t]
	\color{black}
    \begin{align} \label{eq_L_kl}
    \mathcal{L}_{\text{KL}} & =  \textstyle D_{\text{KL}}\big(p(\vec{\mathfrak{z}}_L|\mathbf{Z}_{\rm rd}) \| q(\vec{\mathfrak{z}}_L|\mathbf{Z}_{\rm rd}^{\rm prbs})\big)  +  \sum_{l=1}^{L-1} D_{\text{KL}}\big(q(\vec{\mathfrak{z}}_l|\vec{\mathfrak{z}}_{l-1}) \| p(\vec{\mathfrak{z}}_l|\vec{\mathfrak{z}}_{l+1})\big) \\
    & = \textstyle   \frac{1}{2}\Big( \log\frac{\sigma_{ \mathbf{Z}_{\rm rd}^{\rm prbs}}^2}{\sigma_{\mathbf{Z}_{\rm rd}}^2} + \frac{\sigma_{\mathbf{Z}_{\rm rd}}^2 + (\mu_{\mathbf{Z}_{\rm rd}} - \mu_{\mathbf{Z}_{\rm rd}^{\rm prbs}})^2}{\sigma_{\mathbf{Z}_{\rm rd}^{\rm prbs}}^2} - 1\Big) + \frac{1}{2}\sum_{l=1}^{L-1}\Big(  \log\frac{\sigma_{\vec{\mathfrak{z}}_{l+1}}^2}{\sigma_{\vec{\mathfrak{z}}_{l-1}}^2}  +  \frac{\sigma_{\vec{\mathfrak{z}}_{l-1}}^2  +  (\mu_{\vec{\mathfrak{z}}_{l-1}}  -  \mu_{\vec{\mathfrak{z}}_{l+1}})^2}{\sigma_{\vec{\mathfrak{z}}_{l+1}}^2}  -  1\Big) \nonumber 
\end{align}
\hrule
\end{figure*}

\begin{algorithm}[t]
	\color{black}
	\caption{RDPDNet Training Process}
	\label{alg_pdnet_train}
	\begin{algorithmic}[1]
		\STATE \textbf{Input:} $\mathbf{Z}_{\rm rd}$, $\mathbf{Z}^{\rm prbs}_{\rm rd}$.
		\STATE \textbf{Output:} $\mathbf{\Phi}_{\rm rd}$, and $\mathbf{\Phi}_{\rm mask}$.
		\STATE Initialize $\mathbf{\Phi}_{\rm rd}$, and $\mathbf{\Phi}_{\rm mask}$.
		\REPEAT
		\STATE \textit{\textbf{Stage 1: Update $\mathbf{\Phi}_{\rm rd}$ (fix $\mathbf{\Phi}_{\rm mask}$)}}
		\STATE Compute $\mathbf{\hat{Z}}_{\rm rd} = f_{\rm den}^{\rm train}(\mathbf{Z}_{\rm rd}, \mathbf{Z}_{\rm rd}^{\rm prbs}; \mathbf{\Phi}_{\rm rd})$ from \eqref{eq_train_denoising_process}.
		\STATE Compute $\mathbf{Z}_{\rm rd}^{\rm mix}$ from \eqref{eq_Z_mix}.
		\STATE Compute $\mathbf{\hat{Z}}_{\rm rd}^{\rm mix} = f_{\rm den}^{\rm train}(\mathbf{Z}_{\rm rd}^{\rm mix}, \mathbf{Z}_{\rm rd}^{\rm prbs}; \mathbf{\Phi}_{\rm rd})$ from \eqref{eq_train_denoising_process}.
		\STATE Update $\mathbf{\Phi}_{\rm rd}$ by minimizing $\mathcal{L}_{\rm den}$ in \eqref{eq_L_denoise}.

		\STATE \textit{\textbf{Stage 2: Update $\mathbf{\Phi}_{\rm mask}$ (fix $\mathbf{\Phi}_{\rm rd}$)}}
		\STATE Compute $\mathbf{\hat{Z}}_{\rm rd} = f_{\rm den}^{\rm infer}(\mathbf{Z}_{\rm rd}; \mathbf{\Phi}_{\rm rd})$ from \eqref{eq_infer_denoising_process}.
		\STATE Compute $\mathbf{M}_{\rm mask} = f_{\rm mask}(\mathbf{\hat{Z}}_{\rm rd}; \mathbf{\Phi}_{\rm mask})$.
		\STATE Compute $\mathbf{Z}_{\rm rd}^{\rm mix}$ from \eqref{eq_Z_mix}.
		\STATE Compute $\mathbf{\hat{Z}}_{\rm rd}^{\rm mix} = f_{\rm den}^{\rm infer}(\mathbf{Z}_{\rm rd}^{\rm mix}; \mathbf{\Phi}_{\rm rd})$  from \eqref{eq_infer_denoising_process}.
		\STATE Update $\mathbf{\Phi}_{\rm mask}$ by minimizing $\mathcal{L}_{\rm mask}$ in \eqref{eq_L_mask}.
		\UNTIL{convergence}
	\end{algorithmic}
\end{algorithm}

Distinguishing real target peaks from the spurious or ambiguous peaks caused by random ACF sidelobe fluctuations in the data slots, and from Gaussian thermal noise, is difficult. The AFM method \cite{Ryou2024} is effective here as a training constraint that exploits global frequency-domain (here, RD-matrix) information. Let $\mathbf{M}_{\rm mask} = f_{\rm mask}(\mathbf{\hat{Z}}_{\rm rd}; \mathbf{\Phi}_{\rm mask}) \in [0, 1]^{2 \times N_{\rm chip} \times N_{\rm prbs}}$ be the adversarial weight mask, produced from the denoised output $\mathbf{\hat{Z}}_{\rm rd} = f_{\rm den}(\mathbf{Z}_{\rm rd}; \mathbf{\Phi}_{\rm rd})$ with learnable parameters $\mathbf{\Phi}_{\rm mask}$. To keep true detections intact, a target-support mask $\mathbf{M}_{\rm tars} \in [0, 1]^{2 \times N_{\rm chip} \times N_{\rm prbs}}$ is obtained by normalizing the clean RD power $\mathbf{P}_{\rm rd}^{\rm prbs} = |\mathbf{Z}_{\rm rd}^{\rm prbs}|^2$ so that target peaks map to unity and noise regions to zero, and the adversary is confined to the non-target region via the constrained mask $\mathbf{M}_{\rm afm} = \mathbf{M}_{\rm mask} \odot (1 - \mathbf{M}_{\rm tars})$. Following \cite[Eq. 4, 7]{Ryou2024}, we synthesize an augmented RD matrix $\mathbf{Z}_{\rm rd}^{\rm mix} \in \mathbb{R}^{2 \times N_{\rm chip} \times N_{\rm prbs}}$ to serve as a challenging training input. Using the denoised output $\mathbf{\hat{Z}}_{\rm rd}$,
the objective is to ensure that the targets remain undistorted while the background noise is replaced by a worst-case distribution optimized by the adversary. The mixed RD matrix is formulated as
\begin{align} \label{eq_Z_mix}
    \mathbf{Z}_{\rm rd}^{\rm mix} = \mathbf{M}_{\rm afm} \odot \mathbf{\hat{Z}}_{\rm rd} + (1 - \mathbf{M}_{\rm afm}) \odot \mathbf{Z}_{\rm rd}.
\end{align}
Applying the RDPDNet, the estimated RD matrix is $\mathbf{\hat{Z}}_{\rm rd}^{\rm mix} = f_{\rm den}(\mathbf{Z}_{\rm rd}^{\rm mix}; \mathbf{\Phi}_{\rm rd})$. 
The training process follows a minimax optimization objective. First, we minimize the denoising loss that ensure high-quality reconstruction of true targets on RD matrix, the RDPDNet loss $\mathcal{L}_{\rm den}$ combines global Mean Squared Error (MSE) of the estimated RD matrix and mixed RD matrix with a weighted target-preservation term
\label{anchor_comment_2_m5}
\label{anchor_comment_4_7}
\begin{align} \label{eq_L_rec}
    &\mathcal{L}_{\rm rec} = \|\mathbf{\hat{Z}}_{\rm rd} - \mathbf{Z}_{\rm rd}^{\rm prbs}\|^2_F + \lambda_{\rm rec} \|\mathbf{M}_{\rm tars} \odot (\mathbf{\hat{Z}}_{\rm rd} - \mathbf{Z}_{\rm rd}^{\rm prbs})\|^2_F \\
    &\mathcal{L}^{\rm mix}_{\rm rec} = \|\mathbf{\hat{Z}}_{\rm rd}^{\rm mix} - \mathbf{Z}_{\rm rd}^{\rm prbs}\|^2_F + \lambda_{\rm rec} \|\mathbf{M}_{\rm tars} \odot (\mathbf{\hat{Z}}_{\rm rd}^{\rm mix} - \mathbf{Z}_{\rm rd}^{\rm prbs})\|^2_F
\end{align}
where $\lambda_{\rm rec} \geq 0$ is a predefined target-region weighting coefficient that controls the penalty on reconstruction error within the target support $\mathbf{M}_{\rm tars}$. Combined with the KL loss $\mathcal{L}_{\rm KL}$ in \cite[Eq.~14]{CCECE}, as given in \eqref{eq_L_kl}, the denoising loss function is proposed as
\begin{align} \label{eq_L_denoise}
    \mathcal{L}_{\rm den} = \mathcal{L}_{\rm rec}+\lambda_{\rm mix}\mathcal{L}^{\rm mix}_{\rm rec}+\lambda_{\rm KL}\mathcal{L}_{\rm KL},
\end{align}
where $\lambda_{\rm mix} > 0$ weights the adversarial-mixup reconstruction term $\mathcal{L}^{\rm mix}_{\rm rec}$ relative to the nominal reconstruction term $\mathcal{L}_{\rm rec}$, and thereby sets how strongly the denoiser is trained against the worst-case interference synthesized by the adversary, while $\lambda_{\rm KL} > 0$ weights the latent-divergence term $\mathcal{L}_{\rm KL}$ that enforces the constraint \eqref{eq_problem_2_kl}. Both are fixed hyperparameters, with the values reported in the code repository \cite{github_repo}.
Simultaneously, the mask generator $f_{\rm mask}$ is optimized to maximize the denoiser's reconstruction error while maintaining mask plausibility through a regularization term, we have
\begin{align} \label{eq_L_mask}
    \mathcal{L}_{\rm mask} = -\|\mathbf{\hat{Z}}_{\rm rd}^{\rm mix} - \mathbf{Z}_{\rm rd}^{\rm prbs}\|^2_F + \lambda_{\rm mask} \|\mathbf{M}_{\rm mask}\|^2_F,
\end{align}
where $\lambda_{\rm mask} \geq 0$ is a predefined regularization coefficient that penalizes the energy of the generated mask $\mathbf{M}_{\rm mask}$ to prevent trivial all-one solutions. Collecting the two objectives, training is the constrained minimax game $\min_{\mathbf{\Phi}_{\rm rd}}\max_{\mathbf{\Phi}_{\rm mask}}\mathbb{E}\big[\|\mathbf{\hat{Z}}_{\rm rd}^{\rm mix}-\mathbf{Z}_{\rm rd}^{\rm prbs}\|_F^2\big]$ over the denoiser weights $\mathbf{\Phi}_{\rm rd}=\{\mathbf{\Phi}_{\rm enc},\mathbf{\Phi}_{\rm ext},\mathbf{\Phi}_{\rm dec}\}$ and the adversary weights $\mathbf{\Phi}_{\rm mask}$: the adversary maximizes this mixed reconstruction error through $\mathcal{L}_{\rm mask}$ in \eqref{eq_L_mask}, while the denoiser minimizes it within the composite loss $\mathcal{L}_{\rm den}$ in \eqref{eq_L_denoise}, under the constraints of $\mathcal{P}_2$ in \eqref{eq_problem_2}. This nonconvex problem is solved by the alternating best-response updates of Algorithm~\ref{alg_pdnet_train}, which converge to a stationary point $(\mathbf{\Phi}_{\rm rd}^{\star},\mathbf{\Phi}_{\rm mask}^{\star})$. Training is supervised on paired maps $(\mathbf{Z}_{\rm rd}, \mathbf{Z}_{\rm rd}^{\rm prbs})$ generated in simulation, but the label is not an oracle: $\mathbf{Z}_{\rm rd}^{\rm prbs}$ is the MF output of the data-free MLS over the same scene, obtainable in practice by pilot-only sounding. At inference, $\mathbf{\hat{Z}}_{\rm rd} = f_{\rm den}^{\rm infer}(\mathbf{Z}_{\rm rd}; \mathbf{\Phi}_{\rm rd}) \in \mathbb{R}^{2 \times N_{\rm chip} \times N_{\rm prbs}}$, converted to the complex RD matrix $\mathbf{\hat{Z}}_{\rm rd} \in \mathbb{C}^{N_{\rm chip} \times N_{\rm prbs}}$.

\subsection{Sensing Performance Analysis}
\label{anchor_comment_1_4}
\label{anchor_comment_E_4}
\label{anchor_comment_E_1b}
\subsubsection{Root Mean Squared Error}
To evaluate the estimation accuracy of the proposed method, we adopt the root mean squared error (RMSE), defined for $\theta_m \in \{r_m, v_m\}$ as
\begin{equation}
    \text{RMSE}_\theta = \sqrt{\tfrac{1}{N_{\rm tars}} \textstyle\sum_{m=1}^{N_{\rm tars}} \big(\theta_m - \hat{\theta}_m\big)^2},
\end{equation}
where $r_m$ and $v_m$ are the true range and velocity of target $m$.
\subsubsection{Cramér–Rao Lower Bound}
\label{sub_sec_crlb}
Denote $\vec{\psi}^{(m)} = [r_m,  v_m]^{T}$ as the estimated range $r_m$ and velocity $v_m$ of moving target $m$, and let
$\vec{\Psi} = \big[(\vec{\psi}^{(1)})^{T},  (\vec{\psi}^{(2)})^{T},  \ldots,  (\vec{\psi}^{(N_{\rm tars})})^{T}\big]^{T}$ be the stacked vector of parameters over all targets. From \eqref{eq_vec_y_sen}, the sensing response signal is at time step $t$ with $N_{\rm chip}$, we descrete the time step to chip level, $y_{\rm sen}(t)$, and it can be presented as 
\begin{align}
    \textstyle y_{\rm sen}(t) = \sum_{m = 1}^{N_{\rm tars}}
       \alpha_{\rm sen}^{(m)}   r_{\rm sen}^{(m)}(t; \vec{\psi}^{(m)})+ n_{\rm sen}(t),
\end{align}
where $n_{\rm sen}(t) \sim \mathcal{CN}(0, \sigma_{\rm sen}^2)$ is noise at chip index level, $\alpha_{\rm sen}^{(m)} = \sqrt{\sigma_{\rm tars}^{(m)} P_{\rm tx} G_{\rm tx} A_{e}^{\rm rx}} a_{\rm rx}(\phi^{(m)}_{\rm AOA}\theta^{(m)}_{\rm ZOA}) a_{\rm tx}(\phi^{(m)}_{\rm AOD}\theta^{(m)}_{\rm ZOD})$ is the determistic complex value at time step $t$ of the target $m$, and $r_{\rm sen}^{(m)}(t; \vec{\psi}^{(m)}) = \sqrt{\mathrm{PL}_{\rm sen}^{(m)}(r_m)}   e^{-j2\pi f_c \tau_{\rm sen}^{(m)}(r_m)}   e^{j2\pi f_{{\rm D},{\rm sen}}^{(m)}(v_m)   t}   d\big( t - \tau_{\rm sen}^{(m)}(r_m) \big)$
is the response signal without noise that explicitly contains the parameter $\vec{\psi}^{(m)} = [r_m, v_m]^{\mathsf{T}}$ of the target $m$ through $\tau_{\rm sen}^{(m)}(r_m) = 2r_m/c$, $f_{{\rm D},{\rm sen}}^{(m)}(v_m) = -2v_m/\lambda$, and $\mathrm{PL}_{\rm sen}^{(m)}(r_m) = 1/\big((4\pi)^2 r_m^4\big)$, as defined in \eqref{eq_H_sen_m}. Discreate time step $t$ to chip duration with $t = iT_{\rm chip}$, we have $y_{\mathrm{sen},i} = \sum_{m = 1}^{N_{\rm tars}+ N_{\rm scats}} \alpha_{\rm sen}^{(m)}   r_{\mathrm{sen},i}^{(m)}(\vec{\psi}^{(m)})+ n_{\mathrm{sen},i}$. Because the co-located Tx/Rx is stationary, the $N_{\rm scats}$ stationary scatterers fall in the zero-Doppler bin that is removed in preprocessing; the likelihood below, and hence the CRLB, are therefore conditioned on the moving-target subspace with the stationary clutter treated as known and cancelled, while its residual range-sidelobe leakage into nonzero-Doppler bins is exactly the term $\mathbf{Z}_{\rm rd}^{\rm ran}$ that RDPDNet suppresses. From \cite[Eq. A.74-A.76]{Richards2022}, the log-likelihood function for a complex Gaussian noise process becomes
\begin{align} \label{eq_ln_y}
    &\ln p(y_{\rm sen}|\vec{\Psi}) \\
    & =  (-1/\sigma_{\rm sen}^{2})
    \textstyle\sum_{i=0}^{N_{\rm fr}-1}     | y_{\mathrm{sen},i}  -  \textstyle\sum_{m = 1}^{N_{\rm tars}} \alpha_{\rm sen}^{(m)}   r_{\mathrm{sen},i}^{(m)}(\vec{\psi}^{(m)})|^2, \nonumber
\end{align}
where $N_{\rm fr} = N_{\rm chip} \times N_{\rm prbs}$, up to an additive constant independent of $\vec{\Psi}$.
\label{anchor_comment_2_m2}
The following conditional CRLB assumes known complex gains, approximately orthogonal target responses in delay--Doppler, proper complex Gaussian noise, and negligible delay information from within-chip transitions. Under these assumptions, the joint FIM is approximated by independent $2\times2$ target blocks. Concretely, this delay--Doppler orthogonality holds when any two targets are resolved, i.e., their separation exceeds one resolution cell, $|r_i - r_j| \gtrsim \Delta r = c T_{\rm chip}/2$ or $|v_i - v_j| \gtrsim \Delta v = \lambda/(2 N_{\rm prbs} T_{\rm prbs})$, so that the cross-correlation between their delay--Doppler signatures is negligible and the off-diagonal blocks $\mathbf{F}^{(i,j)}\approx\vec{0}$ vanish. Following \cite[Chapter 7.2]{Richards2022}, the CRLB is obtained from the Fisher information matrix.
\begin{proposition}[Closed-form CRLB for range and velocity] \label{proposition_crlb}
Given that the CRLB for the variance of an unbiased estimate of the \(m\)-th target is \(\operatorname{Var}(\vec{\psi}^{(m)}) \ge \left[\mathbf{F} \left(\vec{\psi}^{(m)}\right)\right]^{-1}\), where The Fisher Information Matrix (FIM) is defined as $\mathbf{F}(\vec{\Psi}) = \mathbb{E}\left[\left(\partial \ln p(y_{\rm sen}|\vec{\Psi})/\partial \vec{\Psi}\right) \left(\partial \ln p(y_{\rm sen}|\vec{\Psi})/\partial \vec{\Psi}\right)^{T} \right].
$
The CRLB of the range and velocity can be presented as
\begin{align}
  \mathrm{Var}(\hat{r}_m) &\ge \frac{F_{vv}^{(m)}}{F_{rr}^{(m)}F_{vv}^{(m)} - \big(F_{rv}^{(m)}\big)^2} \label{eq_crlb_r}, \\
  \mathrm{Var}(\hat{v}_m) &\ge \frac{F_{rr}^{(m)}}{F_{rr}^{(m)}F_{vv}^{(m)} - \big(F_{rv}^{(m)}\big)^2} \label{eq_crlb_v}.
\end{align}
\end{proposition}
\begin{IEEEproof}
The detailed proof is presented in Appendix~\ref{appx_proof_crlb}.
\end{IEEEproof}
The bound in Proposition~\ref{proposition_crlb} presumes an unbiased estimator, yet the CA-CFAR/sinc grid read-out introduces a systematic quantization bias in practice; we therefore adjust the bound before using it as a benchmark for the realized RMSE.
\begin{remark} \label{remark_biased_crlb}
\label{anchor_comment_2_3}
\label{anchor_comment_2_m3}
\label{anchor_comment_3_1}
\label{anchor_comment_E_3}
\label{anchor_comment_3_6b}
The bounds \eqref{eq_crlb_r}--\eqref{eq_crlb_v} hold for unbiased estimators under \eqref{eq_ln_y}. For a scalar $\theta_m \in \{r_m, v_m\}$ estimated at MC trial $i$, the residual error is
\begin{equation} \label{eq_residual}
    e_m^{(i)} = \hat{\theta}_m^{(i)} - \theta_m,
\end{equation}
its trial-averaged (systematic) part is the bias
\begin{equation} \label{eq_bias}
    \mathrm{Bias}(\hat{\theta}_m) = \mathbb{E}[\hat{\theta}_m] - \theta_m
    \approx \tfrac{1}{N_{\rm trial}} \textstyle\sum_{i=1}^{N_{\rm trial}} e_m^{(i)},
\end{equation}
and the MSE decomposes as
\begin{equation} \label{eq_mse_decomp}
    \mathrm{MSE}(\hat{\theta}_m) = \mathrm{Var}(\hat{\theta}_m) + \mathrm{Bias}(\hat{\theta}_m)^2.
\end{equation}
Since the CRLB bounds only $\mathrm{Var}(\hat{\theta}_m)$, and only when $\mathrm{Bias}=0$, a biased estimator may attain MSE below the unbiased CRLB without contradicting the Cram\'er--Rao theorem; comparisons must therefore account for the bias.
\end{remark}

The latent variables $\vec{\mathfrak{z}}_{1:L}$ therefore enter the CRLB comparison solely through the estimator bias, which is the only channel linking the probabilistic denoising model to the analytical sensing bound.
In contrast, RDPDNet by itself is a denoiser/preprocessor, not an estimator in the classical statistical-estimation-theory sense: it does not output the parameters $(r_m, v_m)$ directly, but rather conditions the RD map so that the downstream detection and localization stages yield more accurate $(\hat{r}_m, \hat{v}_m)$. Accordingly, the proposed estimator is the RDPDNet(+AFM)-denoised RD-map pipeline with CFAR detection and sinc interpolation as the read-out, which we adopt as the estimator to which the CRLB is compared throughout this paper. Formally, this estimator is the composition
\begin{align} \label{eq_proposed_estimator_chain}
    \mathbf{Y}_{\rm sen}
    &\xrightarrow{\;\text{MF}, \eqref{eq_Z_sen_i}\;} \mathbf{Z}_{\rm sen}
    \xrightarrow{\;\mathcal{F}_{\rm shift}, \eqref{eq_Z_rd_k_i}\;} \mathbf{Z}_{\rm rd}
    \xrightarrow{\;f_{\rm den}^{\rm infer}, \eqref{eq_infer_denoising_process}\;} \hat{\mathbf{Z}}_{\rm rd} \nonumber \\
    &\xrightarrow{\;\text{CA-CFAR}\;} (k^{(m)}, i^{(m)})
    \xrightarrow{\;\text{sinc}\;} (\hat{r}_m, \hat{v}_m),
\end{align}
where the MF and shifted DFT $\mathcal{F}_{\rm shift}$ form the noisy RD map $\mathbf{Z}_{\rm rd}$, the learned inference map $f_{\rm den}^{\rm infer}(\cdot;\mathbf{\Phi}_{\rm rd})$ denoises it into $\hat{\mathbf{Z}}_{\rm rd}$, CA-CFAR extracts the integer peak $(k^{(m)}, i^{(m)})$, and sinc interpolation refines it to the sub-bin estimate $(\hat{r}_m, \hat{v}_m)$. The only learned block in \eqref{eq_proposed_estimator_chain} is $f_{\rm den}^{\rm infer}$; the remaining blocks are fixed, deterministic signal-processing operators. Because the final read-out reads $(\hat{r}_m, \hat{v}_m)$ off a discretized RD grid, the finite resolutions $\Delta r$ and $\Delta v$ introduce a grid-quantization bias that does not vanish with increasing SNR. Assuming this bias varies slowly with respect to the true parameters, we approximate the variance lower bound by the unbiased CRLB in \eqref{eq_mse_decomp}, so that
\begin{align} \label{eq_biased_bound}
    \mathrm{RMSE} &\gtrsim \sqrt{\mathrm{CRLB}_{\rm unbiased} + \mathrm{Bias}^2}.
\end{align}
In the numerical results, we therefore compare the simulated RMSE of the proposed estimator \eqref{eq_proposed_estimator_chain} against the bias-adjusted benchmark $\sqrt{\mathrm{CRLB}_{\rm unbiased} + \mathrm{Bias}^2}$.

\subsubsection{Evaluating the Bias of the Proposed Estimator}
\label{subsub_bias_eval}
Since the proposed estimator \eqref{eq_proposed_estimator_chain} is not constrained to be unbiased, we characterize biased estimator with bias $b(\theta_m)=\mathbb{E}[\hat{\theta}_m]-\theta_m$ and derivative $b'(\theta_m)=\partial b/\partial\theta_m$, the Cram\'er--Rao inequality generalizes to the \emph{biased CRLB} \cite[Ch.~3]{Kay1993}, we have
\begin{align} \label{eq_biased_crlb}
    \mathrm{Var}(\hat{\theta}_m) &\ge \frac{\big(1+b'(\theta_m)\big)^2}{F_{\theta\theta}(\vec{\psi}^{(m)})},
    \\
    \mathrm{MSE}(\hat{\theta}_m) &\ge \frac{\big(1+b'(\theta_m)\big)^2}{F_{\theta\theta}(\vec{\psi}^{(m)})} + b(\theta_m)^2,
\end{align}
where $F_{\theta\theta}$ is the Fisher information for $\theta_m\in\{r_m,v_m\}$. The bias-adjusted benchmark \eqref{eq_biased_bound} is the special case of \eqref{eq_biased_crlb} under the slowly-varying-bias approximation $b'(\theta_m)\approx 0$. Because the denoiser $f_{\rm den}^{\rm infer}$ injects a data-driven prior learned from paired RD maps, its estimate is biased by construction. The grid-quantization component of $b(\theta_m)$ is a deterministic function of the fractional offset $\delta$ of $\theta_m$ from its nearest bin center, and persists as $\mathrm{SNR}\to\infty$, unlike the variance term.

\subsection{Numerical Results for Range-Doppler Denoising} \label{sec_result}
\label{anchor_comment_2_5}
\label{anchor_comment_2_m9}
\label{anchor_comment_3_2}
\label{anchor_comment_E_5c}

All detailed settings and supplementary analysis are available in our code repository \cite{github_repo}.
The range and Doppler resolutions are fixed at $\Delta r = cT_{\rm chip}/2$ and $\Delta v = \lambda/(2N_{\rm prbs}T_{\rm prbs})$, so RDPDNet enhances detection reliability, not resolution. Fig.~\ref{fig_rd_maps} shows, for one scene at high SNR, the data-carrying map $\mathbf{Z}_{\rm rd} = \mathbf{Z}_{\rm rd}^{\rm prbs} + \mathbf{Z}_{\rm rd}^{\rm ran} + \mathbf{N}_{\rm rd}^{\rm prbs}$, the raw data-free MLS map $\mathbf{Z}_{\rm rd} = \mathbf{Z}_{\rm rd}^{\rm prbs} + \mathbf{N}_{\rm rd}^{\rm prbs}$, and the RDPDNet output $\mathbf{\hat{Z}}_{\rm rd}$. Data embedding raises the sidelobe floor $\mathbf{Z}_{\rm rd}^{\rm ran}$ so the targets are barely visible in either raw map, whereas RDPDNet suppresses the background by tens of dB and preserves all six target signatures at their true range--velocity coordinates, without knowledge of the embedded symbols.
\begin{figure}[htp!]
\centering
\color{black}
    \includegraphics[width=\linewidth]{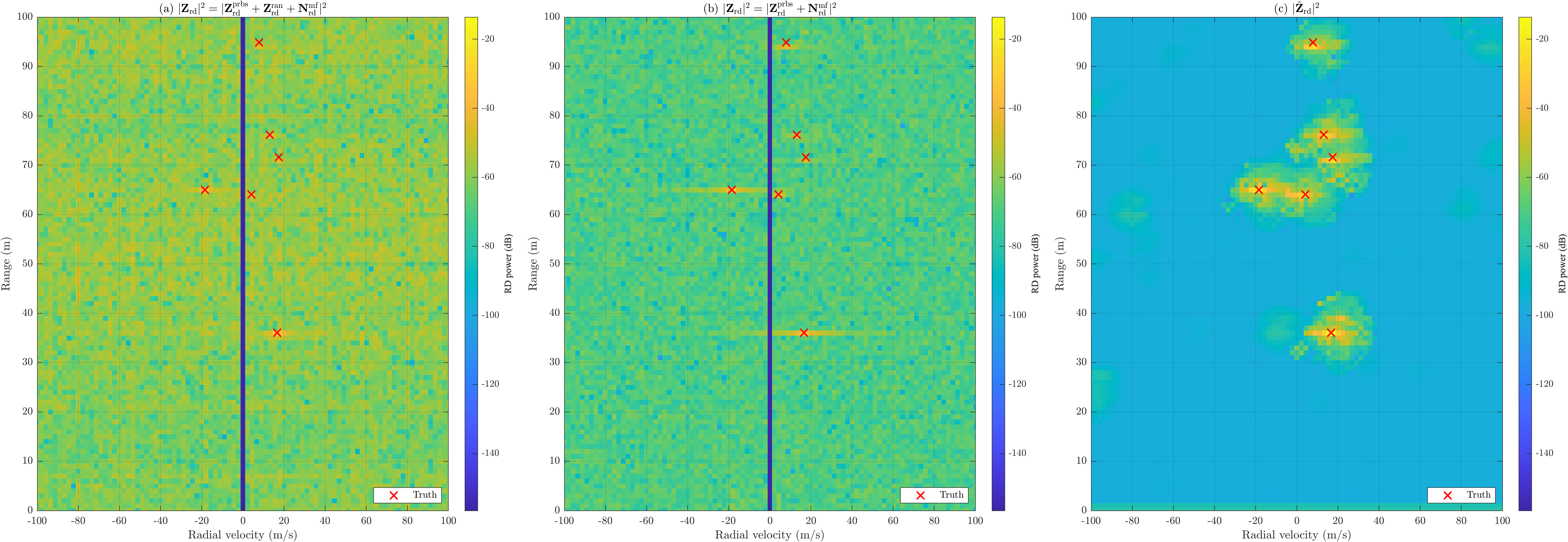}
\caption{RD heatmaps of one scene at $\mathrm{SNR}=\SI{30}{dB}$: (a) data-carrying waveform, (b) raw data-free MLS, (c) RDPDNet-denoised. Crosses mark true targets.}
\label{fig_rd_maps}
\end{figure}
Fig.~\ref{fig_afm_masks} visualizes the AFM masks for the same scene: $\mathbf{M}_{\rm mask}$ floods the low-energy non-target regions, $\mathbf{M}_{\rm tars}$ concentrates on the true targets and the zero-Doppler clutter column, and their combination $\mathbf{M}_{\rm afm} = \mathbf{M}_{\rm mask} \odot (1-\mathbf{M}_{\rm tars})$ perturbs only clutter- and noise-dominated cells, keeping RDPDNet consistent around the target peaks.
\begin{figure}[htp!]
	\centering
	\color{black}
    \includegraphics[width=\linewidth]{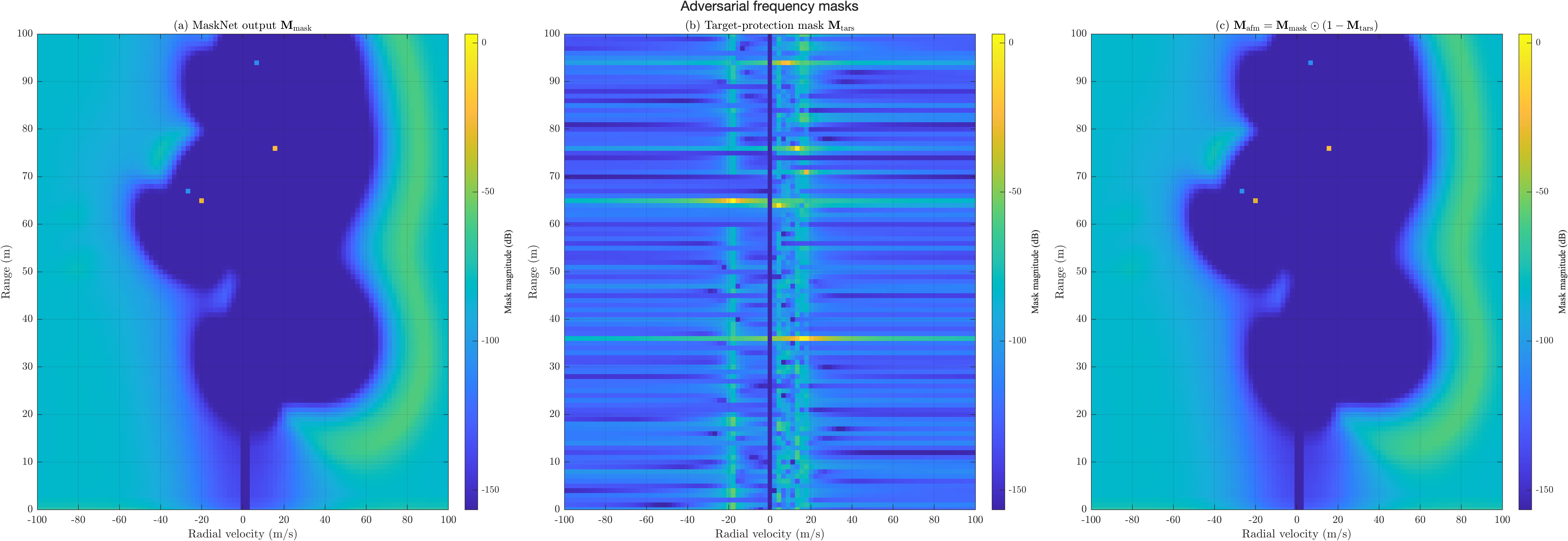}
	\caption{AFM masks for the scene of Fig.~\ref{fig_rd_maps}: (a) MaskNet output $\mathbf{M}_{\rm mask}$, (b) target-protection mask $\mathbf{M}_{\rm tars}$, (c) adversarial mask $\mathbf{M}_{\rm afm}$.}
	\label{fig_afm_masks}
\end{figure}
Fig.~\ref{fig_rmse_N_chip_vs_N_prbs} isolates the waveform-parameter trade-off using the conventional CFAR read-out. For fixed $T_{\rm prbs}$, increasing $N_{\rm chip}$ refines the range resolution $\Delta r = c T_{\rm chip}/2$: from $N_{\rm chip}=64$ to $512$ the high-SNR range RMSE drops by over an order of magnitude (about \SI{30}{m} to \SI{2}{m}), Fig.~\ref{fig_rmse_range_vs_snr_nchip_sweep}. Dually, increasing $N_{\rm prbs}$ refines the velocity resolution $\Delta \nu = \lambda/(2 N_{\rm prbs} T_{\rm prbs})$: the high-SNR velocity RMSE improves from about \SI{13}{m/s} ($N_{\rm prbs}=64$) to below \SI{2}{m/s} ($N_{\rm prbs}=512$), saturating at its resolution floor beyond $\mathrm{SNR}=\SI{15}{dB}$, Fig.~\ref{fig_rmse_vel_vs_snr_nprbs_sweep}. Since $N_{\rm prbs}$ also sets the symbol count $N_{\rm sym} = N_{\rm prbs} N_{\rm bit}^{\rm prbs}$, these sweeps quantify the sensing side of the sensing--communication trade-off through the explicit design parameters $(N_{\rm chip}, N_{\rm prbs})$, a flexibility that fixed-parameter PMCW designs do not offer.
\begin{figure}[htp!]
	\centering
	\color{black}
    \subfloat[]{\includegraphics[width=0.45\linewidth]{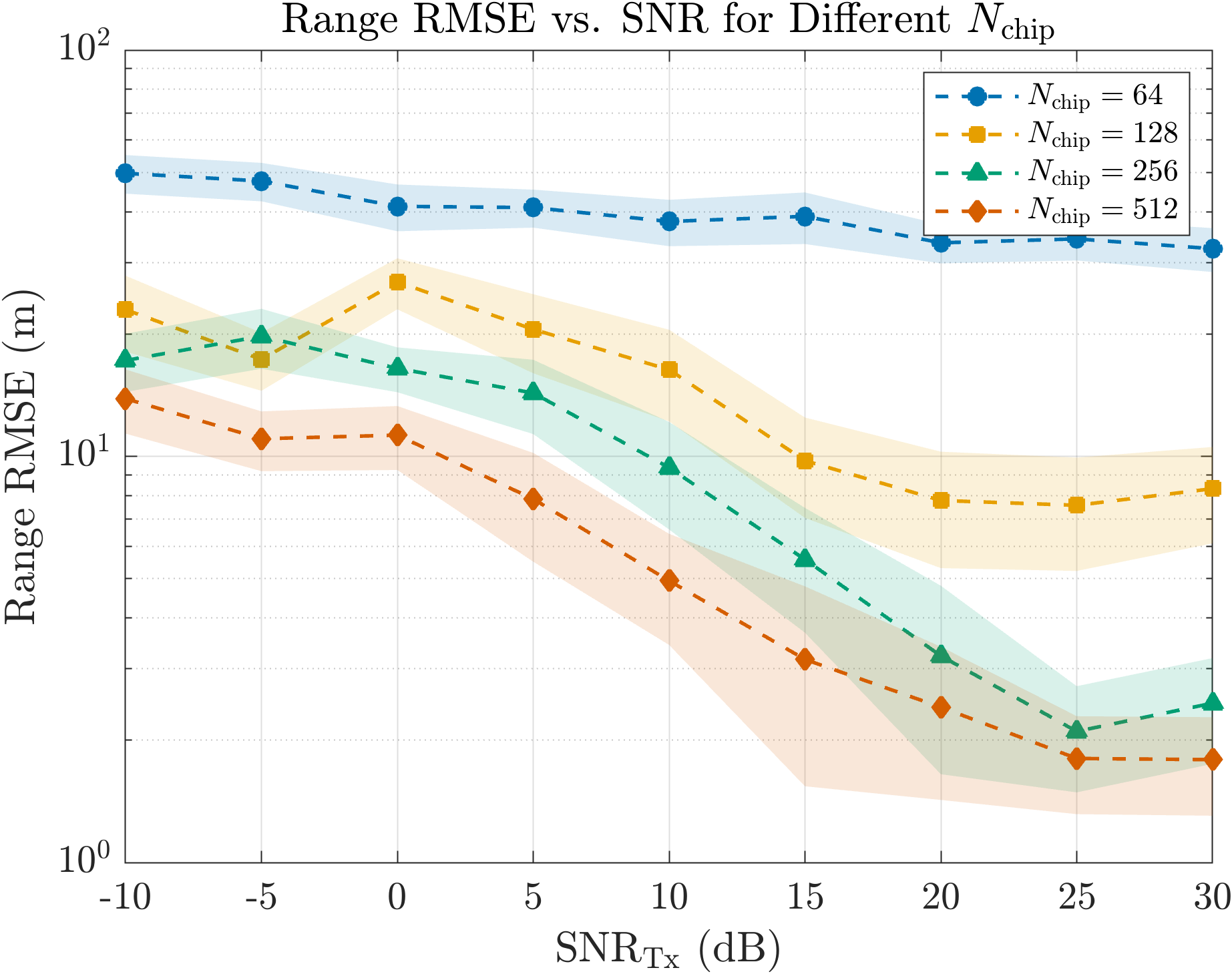}\label{fig_rmse_range_vs_snr_nchip_sweep}}\hfill
    \subfloat[]{\includegraphics[width=0.45\linewidth]{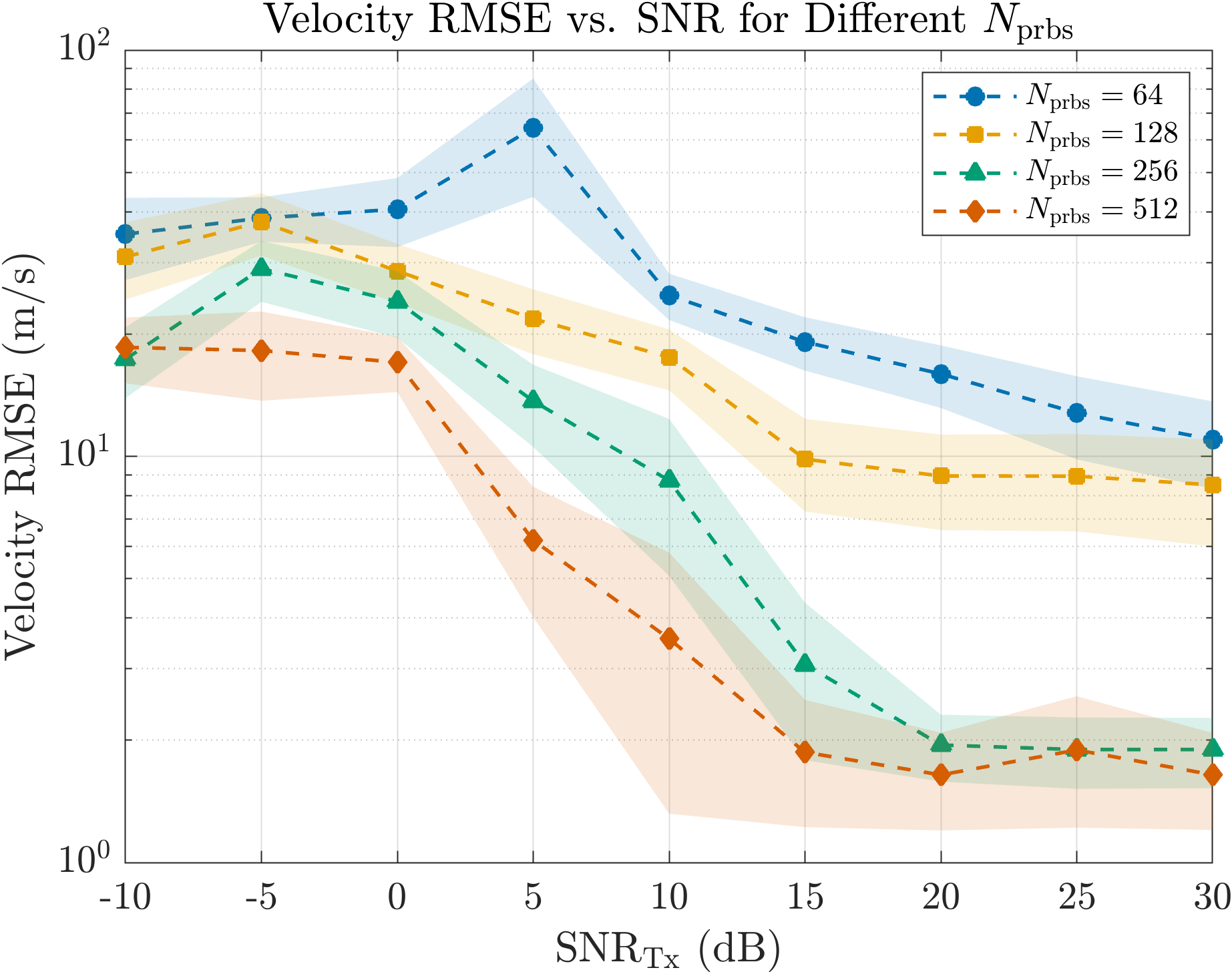}\label{fig_rmse_vel_vs_snr_nprbs_sweep}}\hfill
	\caption{(a) Range RMSE vs.\ SNR for increasing $N_{\rm chip}$ (finer $\Delta r$; $N_{\rm prbs}=256$). (b) Velocity RMSE vs.\ SNR for increasing $N_{\rm prbs}$ (finer $\Delta \nu$; $N_{\rm chip}=512$).}
	\label{fig_rmse_N_chip_vs_N_prbs}
\end{figure}

Fig.~\ref{fig_rmse} evaluates the estimator chains against the closed-form CRLB \eqref{eq_crlb_r}--\eqref{eq_crlb_v}. Fig.~\ref{fig_bias_significance_vs_snr} shows the empirical bias of the conventional data-free pipeline: large and statistically significant below $\mathrm{SNR}=\SI{0}{dB}$ (the detector triggers on noise) and essentially zero beyond $\mathrm{SNR}=\SI{15}{dB}$. This bias is dominated by the CA-CFAR/sinc grid quantization, shared by every chain, so at high SNR all chains inherit the same bias and their RMSE curves converge to a common floor. A single measurement on the data-free pipeline therefore characterizes $b_{\rm bias}$ for all chains and is injected into the benchmark $\sqrt{\mathrm{CRLB} + b_{\rm bias}^2}$ of Remark~\ref{remark_biased_crlb}; at low SNR the biases differ, so this curve is a high-SNR reference, not a per-chain bound.

Figs.~\ref{fig_rmse_range_vs_snr}--\ref{fig_rmse_vel_vs_snr} compare the four chains (RDPDNet-denoised and conventional, each with and without embedded data) sharing the read-out \eqref{eq_proposed_estimator_chain}. Three points follow: (i) embedding data raises the RMSE through the data-dependent sidelobes, a penalty peaking at mid SNR and closing at high SNR, most of which RDPDNet absorbs; (ii) RDPDNet helps most at low SNR, markedly lowering both range and velocity RMSE, with all curves converging to a common grid-quantization floor at high SNR; and (iii) against the bias-adjusted benchmark $\sqrt{\mathrm{CRLB}_{\rm unbiased}+b_{\rm bias}^2}$ of Remark~\ref{remark_biased_crlb}, the near-unbiased conventional chain attains the benchmark at high SNR while the biased RDPDNet chains track it closely, with only a small residual offset from the learned bias. The detection probability stays essentially one, so the conditional RMSE is not inflated by selection effects, and the closed-form CRLB correctly predicts the high-SNR behavior.
\begin{figure*}[!t]
	\centering
	\color{black}
    \subfloat[]{\includegraphics[width=0.48\linewidth]{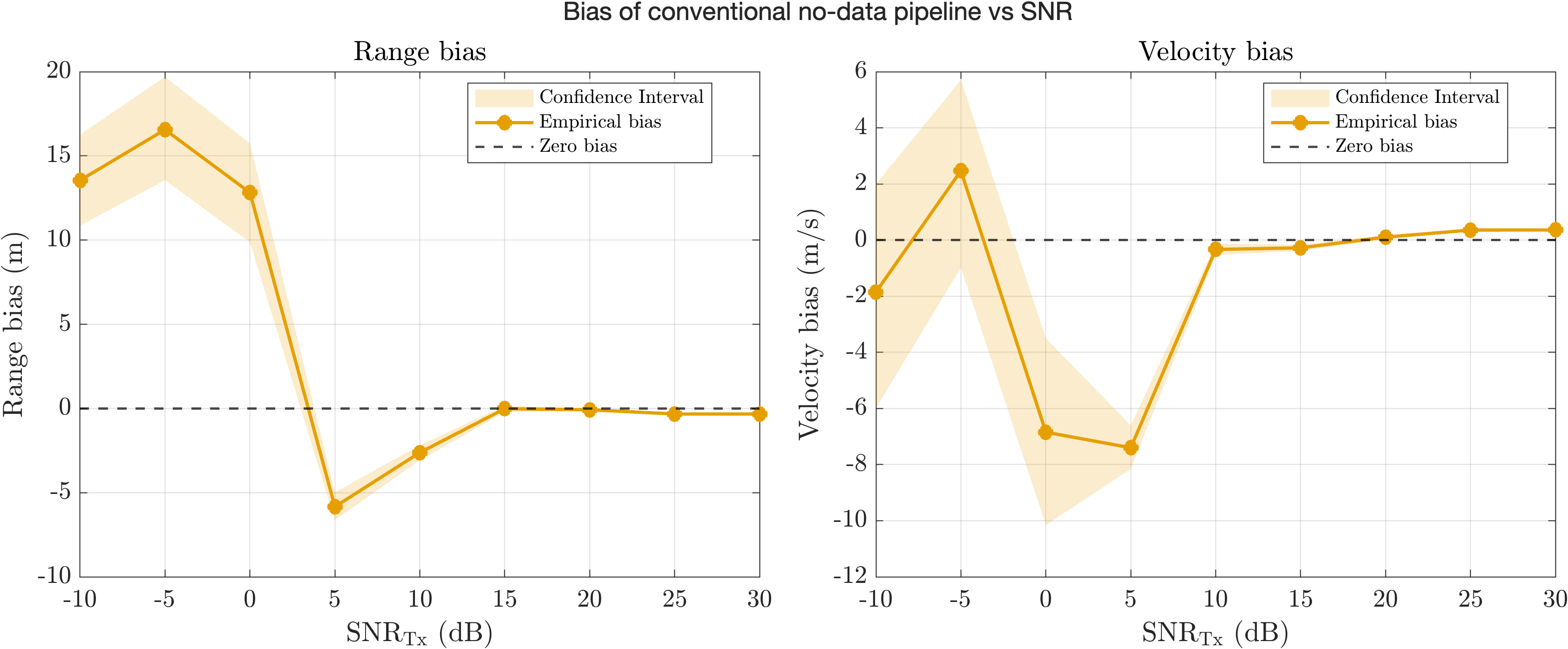}\label{fig_bias_significance_vs_snr}}\hfill
    \subfloat[]{\includegraphics[width=0.24\linewidth]{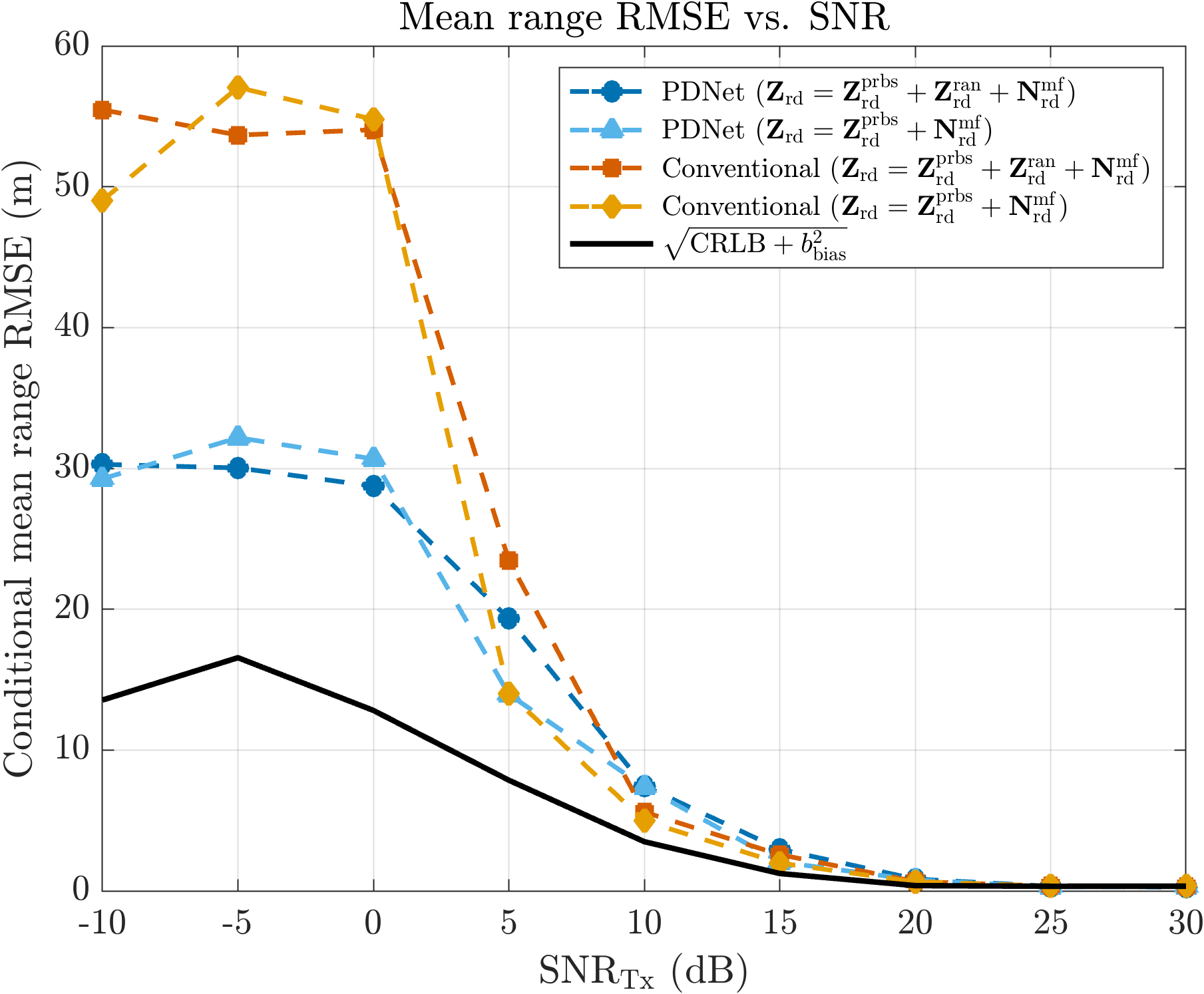}\label{fig_rmse_range_vs_snr}}\hfill
    \subfloat[]{\includegraphics[width=0.24\linewidth]{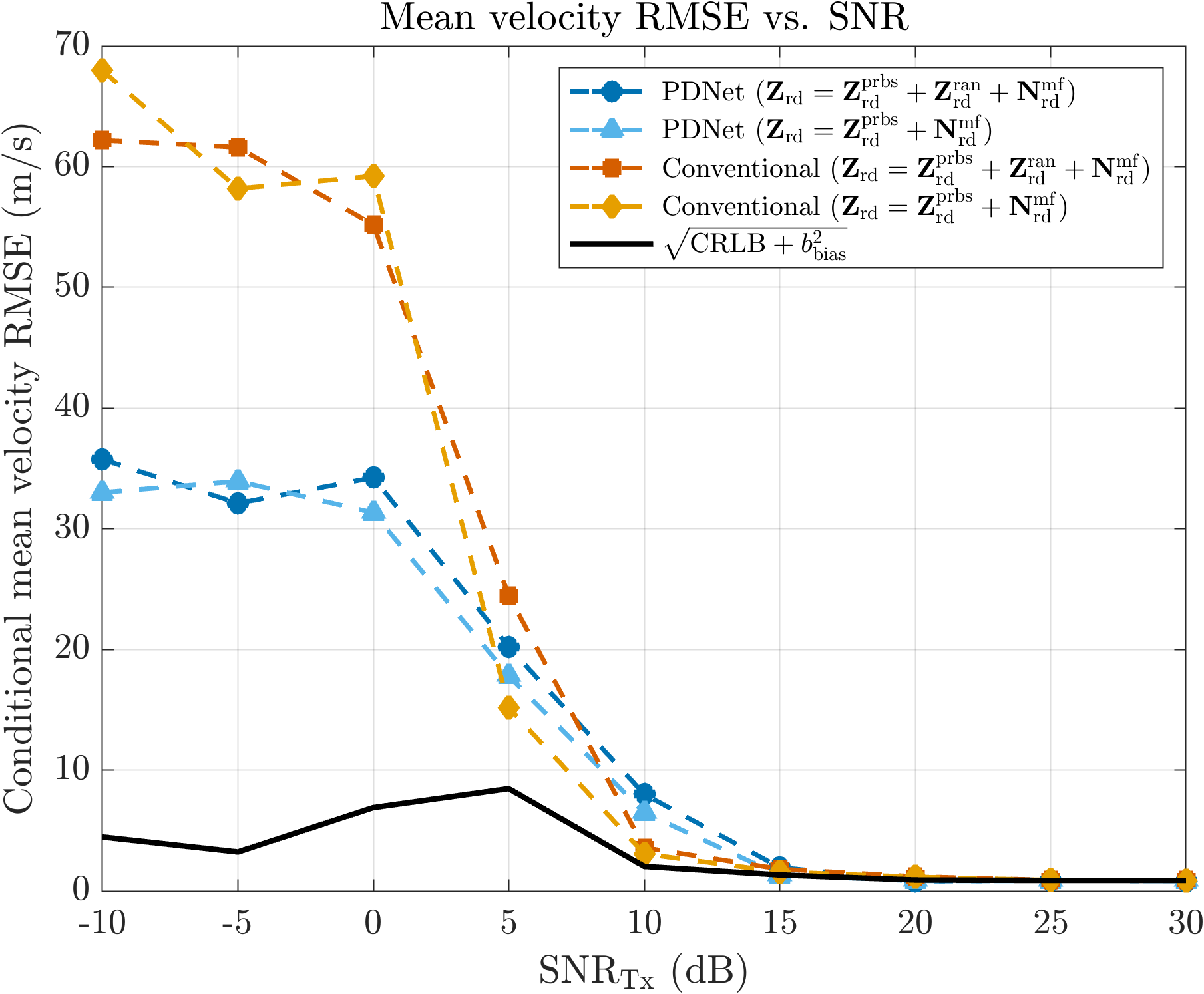}\label{fig_rmse_vel_vs_snr}}\hfill
	\caption{(a) Empirical bias of the conventional pipeline (significant below \SI{0}{dB}, vanishing beyond \SI{15}{dB}). (b) Conditional range and (c) velocity RMSE vs.\ SNR for the RDPDNet-denoised and conventional chains, with/without embedded data, against the benchmark $\sqrt{\mathrm{CRLB}+b_{\rm bias}^2}$.}
	\label{fig_rmse}
\end{figure*}

\label{anchor_comment_1_3}
\label{anchor_comment_4_8}
\label{anchor_comment_E_9}
\label{anchor_comment_E_8b}
Finally, we quantify the computational cost of the added denoising stage to confirm its practical feasibility. Since RDPDNet does not replace any receiver block but is inserted between the RD-map formation and the CA-CFAR detector, the per-frame cost of the proposed chain is the conventional cost plus a single denoiser pass,
\begin{align} \label{eq_complexity_total}
    \underbrace{\mathcal{O} \big(N_{\rm chip} N_{\rm prbs} \log(N_{\rm chip} N_{\rm prbs})\big)}_{\text{MF} + \text{Doppler-FFT} + \text{CA-CFAR}} + \underbrace{\mathcal{O} \big(L_{\rm cnn} \bar{C}^2 k^2 N_{\rm chip} N_{\rm prbs}\big)}_{\text{RDPDNet}},
\end{align}
where $L_{\rm cnn}$ is the number of convolutional layers, $k$ the kernel size, and $\bar{C}$ the average channel width. The denoiser term is linear in the map size $N_{\rm chip} N_{\rm prbs}$ and, unlike the MF stage, is independent of the number of targets and scatterers $N_{\rm tars} + N_{\rm scats}$. At inference, RDPDNet further drops the encoder component and evaluates only $\mathbf{\Phi}_{\rm rd} = \{\mathbf{\Phi}_{\rm ext}, \mathbf{\Phi}_{\rm dec}\}$ in \eqref{eq_infer_denoising_process}, so the added latency of RDPDNet is a single fixed-size forward pass that does not scale with SNR or target count. Table~\ref{tab_complexity} reports the model size, per-frame FLOPs, and measured inference latency of the conventional MF~$+$~CA-CFAR chain (no learnable parameters), the RDPDNet stage alone, and the complete pipeline, on an RD map of size $N_{\rm chip} \times N_{\rm prbs} = 512 \times 256$. The complete pipeline runs in $24.608$ ms per frame, i.e., the conventional $14.880$ ms plus the $9.728$ ms denoiser pass, so the enhancement is obtained at a bounded, additive cost that keeps the receiver real-time capable.
\begin{table*}[!t]
\centering
\color{black}
\caption{Per-frame complexity, model size, and inference latency for a $2\times512\times256$ input. RDPDNet is a preprocessing stage inserted into the conventional chain, so the third row is the sum of the first two; hardware and batch size are in \cite{github_repo}.}
\label{tab_complexity}
\label{anchor_comment_2_m15}
\setlength{\tabcolsep}{6pt}
\renewcommand{\arraystretch}{1.2}
\color{black}
\begin{tabular}{lcccc}
\hline
Method & Order complexity & Params & FLOPs/frame & Latency \\
\hline
Conventional chain: MF $+$ CA-CFAR & $\mathcal{O}(N_{\rm chip} N_{\rm prbs} \log(N_{\rm chip} N_{\rm prbs}))$ & $0$ & $0.028$ G & $14.880$ \si{ms} \\
Inserted stage: \textbf{RDPDNet (ours)} & $\mathcal{O}(L_{\rm cnn} \bar{C}^2 k^2 N_{\rm chip} N_{\rm prbs})$ & $162288$ & $18.280$ G & $9.728$ \si{ms} \\
Proposed pipeline: MF $+$ \textbf{RDPDNet} $+$ CA-CFAR & $\mathcal{O}(L_{\rm cnn} \bar{C}^2 k^2 N_{\rm chip} N_{\rm prbs})$ & $162288$ & $18.308$ G & $24.608$ \si{ms} \\
\hline
\end{tabular}
\end{table*}

\section{Communication Decoding with the PDISAC Waveform} \label{sec_com}

\subsection{Delay Estimation at the UE}
Consider $\mathbf Y_{\rm com}\in\mathbb C^{N_{\rm chip}\times(N_{\rm prbs}+1)}$ from \eqref{eq_Y_com}. The additional received MLS preserves the delayed tail and permits each transmitted MLS to be reconstructed from two consecutive received MLS sequences after synchronization. Delay estimation itself is performed on the current received MLS using the known reference sequence.
With the MLS known at the UE, delay estimation is cast as a ACF between the received MLS $\vec y_{{\rm com},i}$ and the conjugated modulated MLS $\vec p_{{\rm prbs},i}\in\{-1+0j,1+0j\}^{N_{\rm chip}\times1}$, $\vec{\text{ACF}}_{\rm com}(t_i,\tau_{\rm ue})=\vec y_{{\rm com},i}\circledast\vec p_{{\rm prbs},i}$, which can be expanded to
\begin{align}\label{eq_acf_com_delay_b}
\vec{\text{ACF}}_{\rm com}(t_i,\tau_{\rm ue})
=(\vec d_{{\rm isac},i}^{\rm los}\circledast\vec p_{{\rm prbs},i}^{*})h_{{\rm com},i}^{\rm los}+\vec n_{{\rm com},i}^{\rm mf},
\end{align}
where $\vec n_{{\rm com},i}^{\rm mf}$ is the noise vector after the MF.
Since $\vec d_{{\rm isac},i}^{\rm los}$ carries communication data at the even (data) slots, the ACF fluctuates and the delay peak becomes unreliable if computed on the full sequence. The odd (pilot) and even (data) contributions to \eqref{eq_acf_com_delay_b} are therefore evaluated separately and combined by magnitude, giving the delay index and delay
\begin{align}
&\hat k_{\tau,i} \\
& =\arg\max_{k}\Big(\big|[\vec{\text{ACF}}_{\rm com}^{\rm odd}(t_i,\tau_{\rm ue})]_k\big|+\big|[\vec{\text{ACF}}_{\rm com}^{\rm even}(t_i,\tau_{\rm ue})]_k\big|\Big), \nonumber
\end{align}
where $\vec{\text{ACF}}_{\rm com}^{\rm odd}(t_i,\tau_{\rm ue})$ and $\vec{\text{ACF}}_{\rm com}^{\rm even}(t_i,\tau_{\rm ue})$ denote the odd (pilot) and even (data) slot components of \eqref{eq_acf_com_delay_b}, obtained by masking the reference $\vec p_{{\rm prbs},i}$ to its odd and even slots, respectively, before the circular correlation. The odd term is coherent because it contains no unknown data symbol, whereas the magnitude operation prevents sign changes in the even data slots from canceling its delay evidence. The delay estimate is $\hat\tau_{{\rm ue},i}=(\hat k_{\tau,i}-1)T_{\rm chip}$.
The received MLS sequences are concatenated pairwise and shifted by $\hat k_{\tau,i}$ to obtain $\hat{\mathbf Y}_{\rm com}\in\mathbb C^{N_{\rm chip}\times N_{\rm prbs}}$. After alignment to the LOS delay, its NLOS replicas retain the relative delays $\tau_{\rm nlos}^{(m)}-\tau_{\rm ue}$
\begin{align}\label{eq_Y_shifted_com_b}
\hspace{-0.2cm}\hat{\mathbf Y}_{\rm com}=\mathbf D_{\rm isac}\mathbf H_{\rm com}^{\rm los}
\!+\!\sum_{m=1}^{N_{\rm tars}^{\rm ref}}\mathbf D_{\rm isac}^{{\rm nlos},(m)}
\mathbf H_{\rm com}^{{\rm nlos},(m)}+\mathbf N_{\rm com}.
\end{align}

\subsection{Channel Estimation at Pilot Slots}
\label{anchor_comment_2_m1}
\label{anchor_comment_E_10c}
\label{anchor_comment_3_6c}
\label{anchor_comment_4_9c}
After synchronization, each MLS is divided into alternating pilot and data slots of length $N_{\rm chip}^{\rm slot}$. For slot pair $k$ in MLS $i$, let $\vec p_{{\rm prbs},k}^{\rm odd}$ and $\vec p_{{\rm prbs},k}^{\rm even}$ denote the known MLS chip vectors. The pilot observation is
\begin{align}\label{eq_y_com_odd_slot_b}
&\vec y_{{\rm com},i}^{\rm odd,(k)} \nonumber \\
&=\vec p_{{\rm prbs},k}^{\rm odd}h_{{\rm com},i}^{\rm los}
+\sum_m\vec p_{{\rm prbs},k}^{{\rm odd},(m)}h_{{\rm com},i}^{{\rm nlos},(m)}
+\vec n_{{\rm com},k}^{\rm odd}.
\end{align}
The implementation estimates a single scalar channel per slot via least squares (LS)
\begin{align}\label{eq_est_h_com_ik_b}
\hat h_{{\rm com},i}^{(k)}
=\frac{(\vec p_{{\rm prbs},k}^{\rm odd})^{\mathsf{H}}\vec y_{{\rm com},i}^{\rm odd,(k)}}
{(\vec p_{{\rm prbs},k}^{\rm odd})^{\mathsf{H}}\vec p_{{\rm prbs},k}^{\rm odd}}.
\end{align}
This LS estimate is unbiased whenever the NLOS interference is zero-mean and uncorrelated with $\vec p_{{\rm prbs},k}^{\rm odd}$; under a flat, interference-free channel it reduces to $\hat h_{{\rm com},i}^{(k)}=h_{{\rm com},i}^{\rm los}$.
The even-slot observation is more generally written using its delayed data-bearing NLOS replicas
\begin{align}\label{eq_y_com_even_slot_b}
&\vec y_{{\rm com},i}^{\rm even,(k)} \nonumber \\
&=s_{k,i}\vec p_{{\rm prbs},k}^{\rm even}h_{{\rm com},i}^{\rm los}
+\sum_m\vec d_{k,i}^{{\rm even},(m)}h_{{\rm com},i}^{{\rm nlos},(m)}
+\vec n_{{\rm com},k}^{\rm even}.
\end{align}
The delayed vector $\vec d_{k,i}^{{\rm even},(m)}$ need not contain only the current symbol when an NLOS delay crosses a slot or MLS boundary.

\subsection{Data Detection via Channel Equalization}
The data-slot MF statistic is $z_{{\rm com},i}^{(k)}=(\vec p_{{\rm prbs},k}^{\rm even})^{\mathsf{H}}\vec y_{{\rm com},i}^{\rm even,(k)}$. Consistent with the scalar estimator \eqref{eq_est_h_com_ik_b}, the equalized symbol is $\hat s_{k,i}= z_{{\rm com},i}^{(k)}/\big(\hat h_{{\rm com},i}^{(k)} (\vec p_{{\rm prbs},k}^{\rm even})^{\mathsf{H}}\vec p_{{\rm prbs},k}^{\rm even}\big)$. Consistent with the bit-to-symbol mapping $s_n=2b_n-1$, the hard decision is $\Re\{\hat s_{k,i}\}>0\Rightarrow\hat b_{k,i}=1$ and $\Re\{\hat s_{k,i}\}<0\Rightarrow\hat b_{k,i}=0$. Across $N_{\rm prbs}$ MLS sequences, the total number of decoded symbols is $N_{\rm sym}=N_{\rm prbs}N_{\rm bit}^{\rm prbs}$.

\subsection{Bit Error Rate Analysis}
\label{anchor_comment_4_5}
\label{anchor_comment_E_1c}
For a synchronized data slot, let $L=N_{\rm chip}^{\rm slot}$ and $c_m=(\vec p_{{\rm prbs},k}^{\rm even})^{\mathsf{H}}\vec p_{{\rm prbs},k}^{{\rm even},(m)}$. Conditioned on a geometry-determined channel realization, the instantaneous signal-to-interference-plus-noise ratio (SINR) is
\begin{align}\label{eq_sinr_k_i_b}
\gamma_{k,i}=\frac{|h_{{\rm com},i}^{\rm los}|^2L^2}
{\sum_{m=1}^{N_{\rm tars}^{\rm ref}}|h_{{\rm com},i}^{{\rm nlos},(m)}|^2|c_m|^2+\sigma_{\rm com}^2L}.
\end{align}
For a deterministic realization, the exact coherent interference is $|\sum_mh_{{\rm com},i}^{{\rm nlos},(m)}c_m|^2$; \eqref{eq_sinr_k_i_b} neglects its cross terms under the approximation of weakly correlated, approximately uniform geometry-induced phases in Section~\ref{sec_fading}.

Rather than the conservative perfect-overlap convention $|c_m|^2=L^2$, the per-path despread interference $P_{\rm nlos}^{(m)}=|c_m|^2|h_{{\rm com},i}^{{\rm nlos},(m)}|^2$ is modeled as i.i.d.\ across the $N_{\rm tars}^{\rm ref}$ reflected paths with mean $\Omega=\mathbb E[P_{\rm nlos}^{(m)}]$, so the aggregate interference $P_{\rm nlos}=\sum_m P_{\rm nlos}^{(m)}$ follows a Gamma distribution,
\begin{align}
f_{P_{\rm nlos}}(x)=\frac{x^{N_{\rm tars}^{\rm ref}-1}e^{-x/\Omega}}{\Omega^{N_{\rm tars}^{\rm ref}}(N_{\rm tars}^{\rm ref}-1)!},\quad x\ge0.
\end{align}
Defining $P_{\rm los}=|h_{{\rm com},i}^{\rm los}|^2$, the conditional coherent binary phase-shift keying (BPSK) error probability is
\begin{align}\label{eq_ins_p_e_b}
p_e=Q \left(\sqrt{2\gamma_{k,i}}\right)
=Q \left(\sqrt{\frac{2P_{\rm los}L^2}{P_{\rm nlos}+\sigma_{\rm com}^2L}}\right).
\end{align}
$P_{\rm los}$ is deterministic given the scene and time step, while $P_{\rm nlos}\sim\mathrm{Gamma}(N_{\rm tars}^{\rm ref},\Omega)$ is random over the reflector.

\begin{proposition}[Analytical average BER]\label{proposition_ber_b}
Let $p_e^{\min}=Q\big(\sqrt{2P_{\rm los}L/\sigma_{\rm com}^2}\big)$ be the interference-free error floor. The expected average BER over $p_e\in[p_e^{\min},1/2)$ is
\begin{align} \label{eq_avg_ber_b}
\bar P_e
= \int_{p_e^{\min}}^{1/2}
&p_e 
\frac{\big[\Xi(p_e)\big]^{N_{\rm tars}^{\rm ref}-1}
     \exp \big(-\dfrac{\Xi(p_e)}{\Omega}\big)}
     {\Omega^{N_{\rm tars}^{\rm ref}}
      \bigl(N_{\rm tars}^{\rm ref}-1\bigr)!}
\nonumber \\
& \times \frac{4\sqrt{2\pi} P_{\rm los}L^2}
     {\bigl[Q^{-1}(p_e)\bigr]^3}
\exp \big(\frac{\bigl[Q^{-1}(p_e)\bigr]^2}{2}\big)
 dp_e,
\end{align}
where $\Xi(p_e) = 2P_{\rm los}L^2/\bigl[Q^{-1}(p_e)\bigr]^2 - \sigma_{\rm com}^2 L$.
\end{proposition}
\begin{IEEEproof}
The detailed change-of-variables derivation is presented in Appendix~\ref{appx_proof_BER}.
\end{IEEEproof}

\subsection{Numerical Results for Communication Performance}

Fig.~\ref{fig_com_channel_dis} shows the empirical distributions of the geometry-determined channel: the LOS component follows the bimodal arcsine density of a constant-amplitude path with uniform phase (modes at $\pm|h_{\rm com}^{\rm los}|$), while the aggregate NLOS component matches a zero-mean $\mathcal{CN}(0,\Omega)$ fit, consistent with the central-limit argument of Section~\ref{sec_fading}. These descriptive fits validate the geometry-based characterization that underpins the semi-analytical BER and capacity expressions below, without imposing per-path fading on the simulation.

Fig.~\ref{fig_ber_nsym} plots the BER against the number of data bits per sequence $N_{\rm bit}^{\rm prbs}$. A larger $N_{\rm bit}^{\rm prbs}$ shortens the per-symbol slot $L = N_{\rm chip}/(2 N_{\rm bit}^{\rm prbs})$ and hence the despreading gain \eqref{eq_sinr_k_i_b}, so the BER rises monotonically, from about $0.23$ at $N_{\rm bit}^{\rm prbs}=4$ to $0.39$ at $N_{\rm bit}^{\rm prbs}=128$ under perfect $(\tau_{\rm ue}, h_{{\rm com},i})$; here the ordinate is the BER averaged over the SNR sweep (the abscissa is the allocation, not SNR), so it is dominated by low-SNR points and is not the BER at any single SNR. The semi-analytical curves of Proposition~\ref{proposition_ber_b} coincide with the numerical BER, and the partially informed receivers ($\hat{\tau}_{\rm ue}$ and/or $\hat{h}_{{\rm com},i}$) add a small penalty that shrinks as interference dominates.

Fig.~\ref{fig_ber_snr} shows the corresponding BER versus SNR. The analytical predictions of Proposition~\ref{proposition_ber_b} match the MC results at every allocation and SNR, validating the geometry-conditioned BER characterization without any parametric fading assumption.
\begin{figure*}[!t]
\centering
\color{black}
    \subfloat[]{\includegraphics[width=0.235\linewidth]{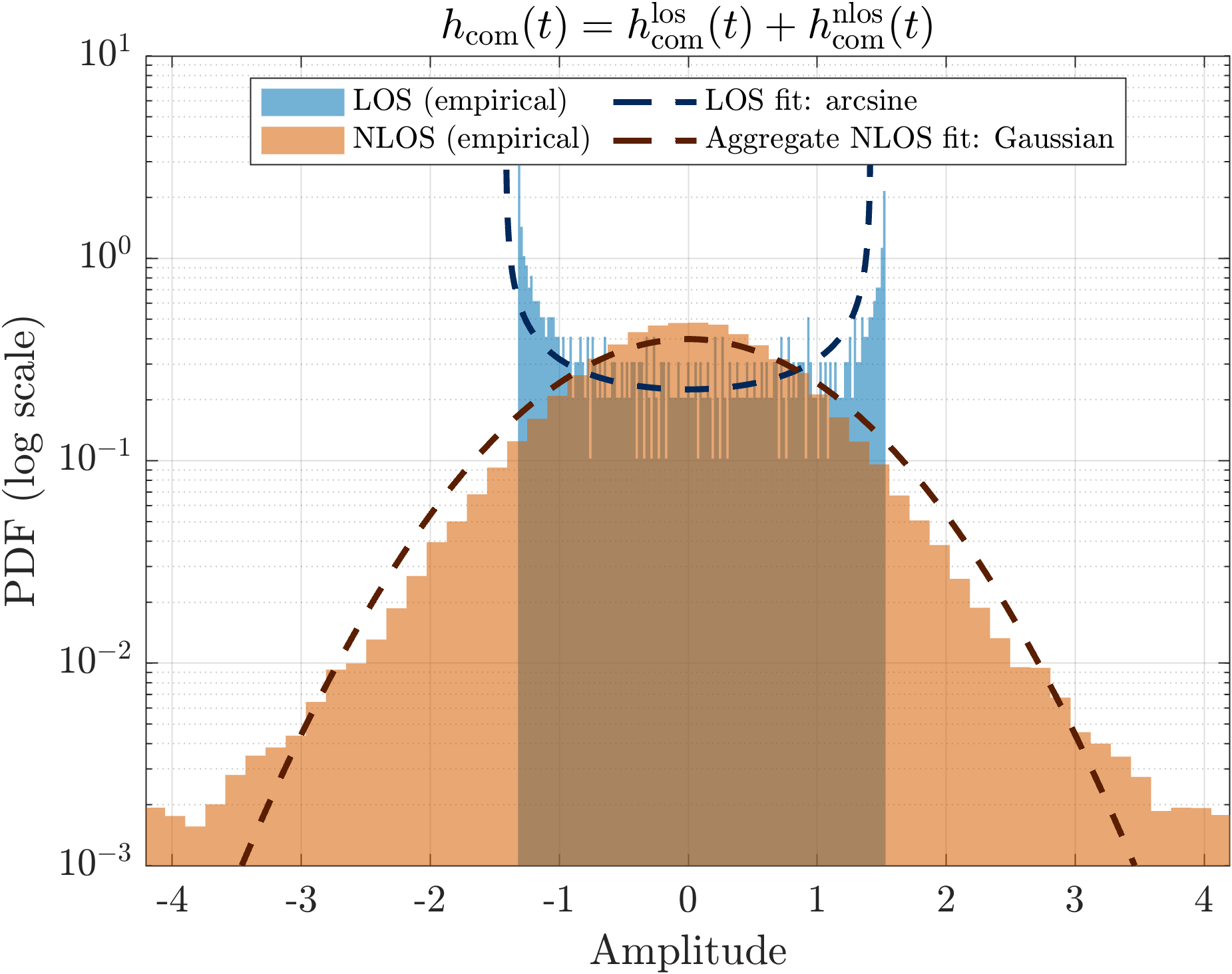}\label{fig_com_channel_dis}}\hfill
    \subfloat[]{\includegraphics[width=0.235\linewidth]{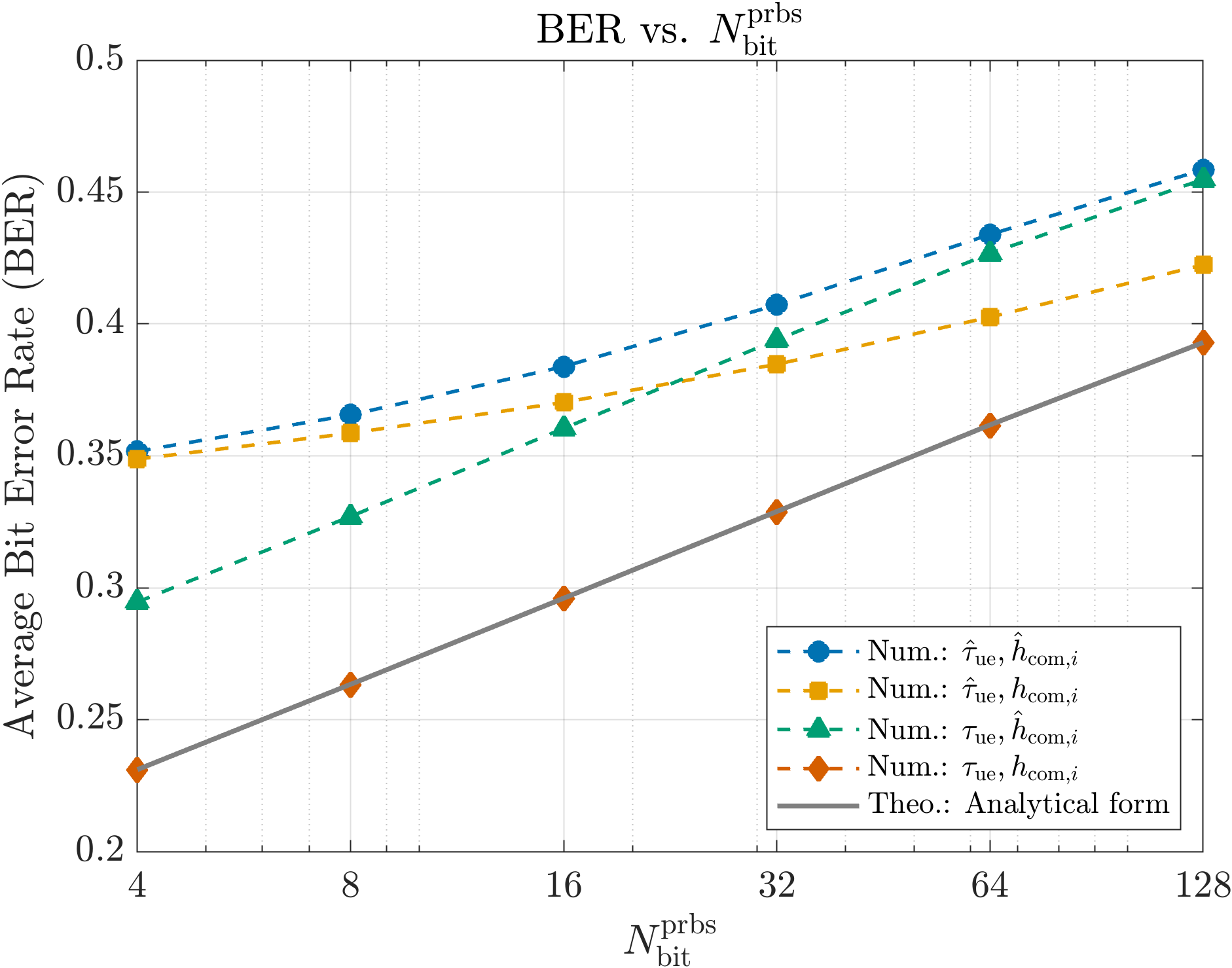}\label{fig_ber_nsym}}\hfill
    \subfloat[]{\includegraphics[width=0.235\linewidth]{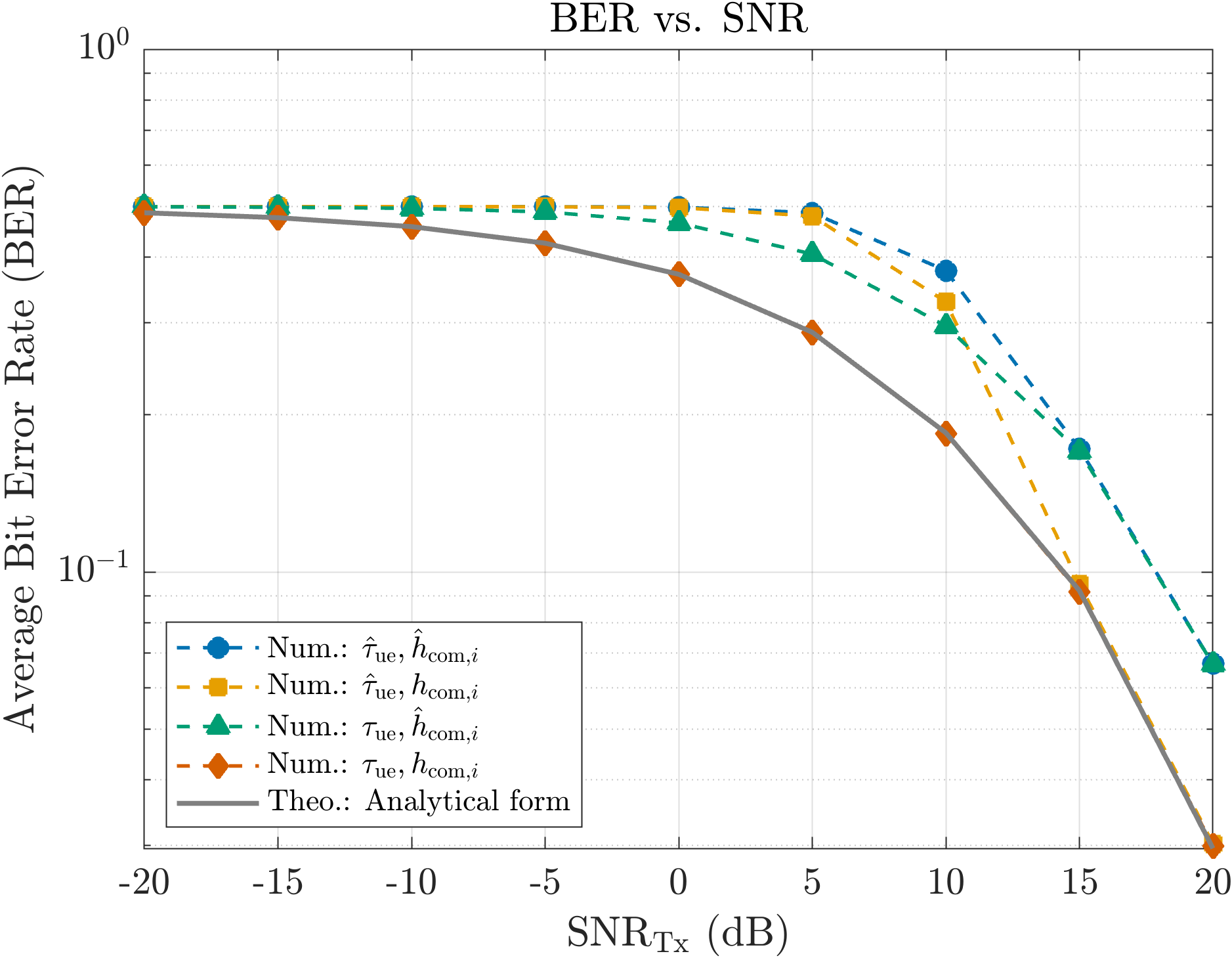}\label{fig_ber_snr}}\hfill
    \subfloat[]{\includegraphics[width=0.24\linewidth]{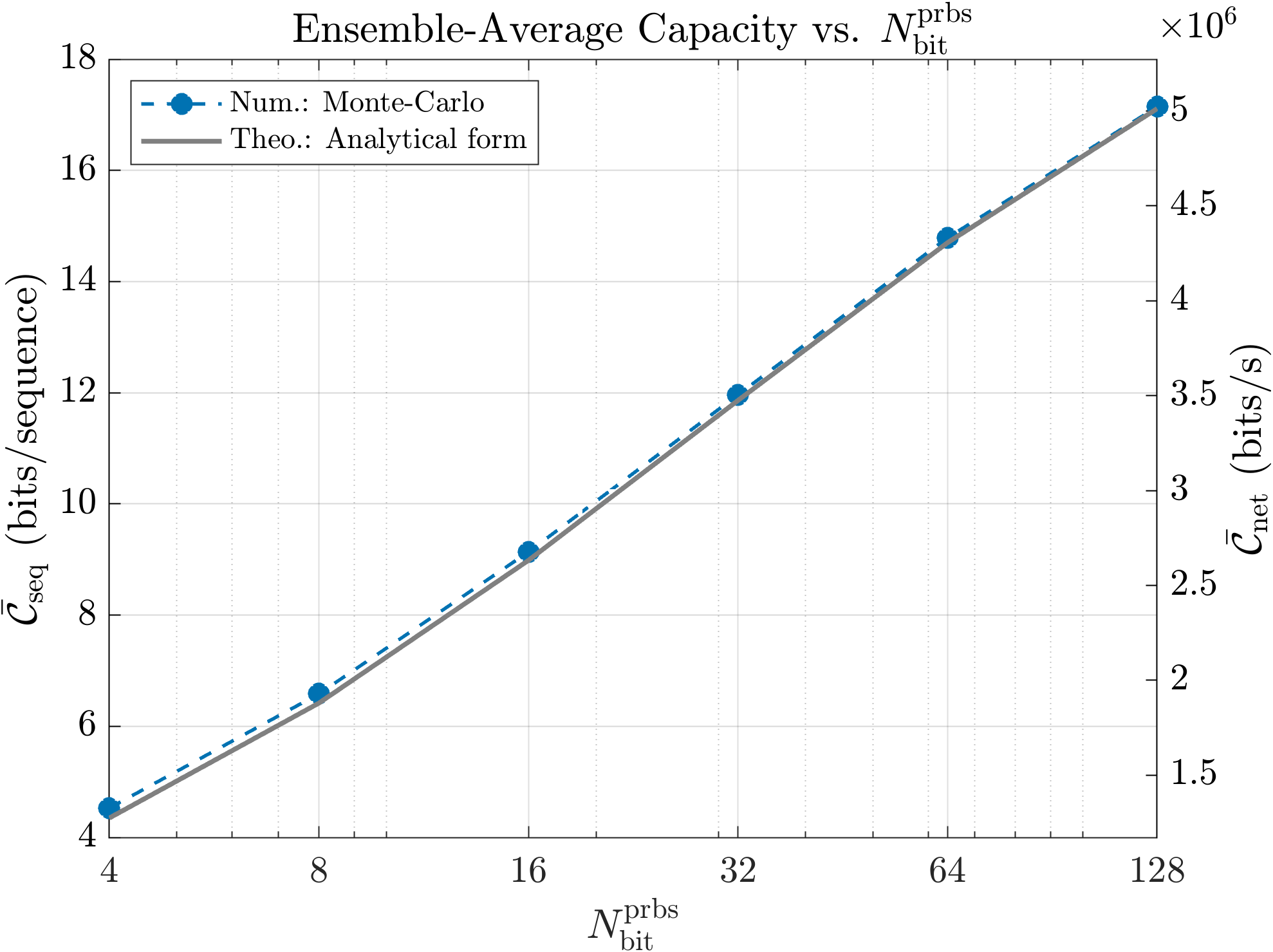}\label{fig_cap_nsym}}
\caption{(a) Empirical LOS and aggregate NLOS channel distributions with descriptive fits (arcsine and $\mathcal{CN}(0,\Omega)$). (b) Raw uncoded BER (SNR-averaged) vs.\ bits per sequence for the four channel state information (CSI)/synchronization conditions. (c) BER vs.\ SNR for different allocations under perfect $(\tau_{\rm ue}, h_{{\rm com},i})$, with the predictions of Proposition~\ref{proposition_ber_b}. (d) Average Shannon capacity per sequence $\bar{\mathcal{C}}_{\rm seq}$ (left axis) and net capacity $\bar{\mathcal{C}}_{\rm sec}$ (right axis) vs.\ $N_{\rm bit}^{\rm prbs}$ \eqref{eq_cap_seq_avg}--\eqref{eq_cap_net}, closed-form (solid) vs.\ MC (markers); pilot-discounted Shannon limits, distinct from the uncoded BER in (b)--(c).}
\label{fig_communication_comparation}
\end{figure*}

\subsubsection*{Average Capacity Analysis}
\label{anchor_comment_2_6}
\label{anchor_comment_2_m13}
\label{anchor_comment_3_3}
\label{anchor_comment_4_4}
\label{anchor_comment_E_6}
We consider the capacity via the Shannon rate $\log_2(1+\gamma_{k,i})$ of the data-bearing slots, which is characterized separately by Figs.~\ref{fig_ber_nsym}--\ref{fig_ber_snr}. Beyond that BER-based throughput, the fundamental limit of the data-carrying slots follows from the Shannon capacity applied to the slot-partitioned waveform of \eqref{eq_d_isac_i} and the detection SINR of \eqref{eq_sinr_k_i_b}. Each MLS sequence of duration $T_{\rm prbs} = N_{\rm chip} T_{\rm chip}$ is split into $N_{\rm slot} = 2N_{\rm bit}^{\rm prbs}$ alternating pilot and data slots: the $N_{\rm bit}^{\rm prbs}$ pilot (sounding) slots are the overhead paid for delay synchronization, channel estimation, and sensing, and together occupy the duration $T_{\rm pilot} = N_{\rm bit}^{\rm prbs} N_{\rm chip}^{\rm slot} T_{\rm chip} = T_{\rm prbs}/2$, so only the fraction $\frac{T_{\rm prbs} - T_{\rm pilot}}{T_{\rm prbs}} = \frac{1}{2}$ of the sequence transports data. After despreading, each data slot of $L = N_{\rm chip}/(2N_{\rm bit}^{\rm prbs})$ chips collapses into one channel use with SINR $\gamma_{k,i}$ given by \eqref{eq_sinr_k_i_b}, so the $i$-th sequence supports at $\mathcal{C}_{\rm seq}^{(i)} = \sum_{k=1}^{N_{\rm bit}^{\rm prbs}} \log_2 \big(1+\gamma_{k,i}\big) \quad \text{[bits/sequence]}$.
Since the sum runs only over the $N_{\rm bit}^{\rm prbs}$ data slots, the pilot overhead is already discounted structurally, and the average overhead-discounted capacity per sequence is obtained by averaging $\mathcal{C}_{\rm seq}^{(i)}$ over the $N_{\rm prbs}$ sequences of the frame,
\begin{align} \label{eq_cap_seq_avg}
	\bar{\mathcal{C}}_{\rm seq}
	= \frac{1}{N_{\rm prbs}} \sum_{i=1}^{N_{\rm prbs}} \sum_{k=1}^{N_{\rm bit}^{\rm prbs}}
	\log_2 \big(1+\gamma_{k,i}\big)
	\quad \text{[bits/seq]}.
\end{align}
The corresponding average net capacity, normalized per sequence duration, $\bar{\mathcal{C}}_{\rm sec}
= \frac{\bar{\mathcal{C}}_{\rm seq}}{T_{\rm prbs}}$, is expanded to
\begin{align} \label{eq_cap_net}
	\hspace{-0.3em}\bar{\mathcal{C}}_{\rm sec}
	=
	\frac{1}{N_{\rm prbs} T_{\rm prbs}}
	\sum_{i=1}^{N_{\rm prbs}} \sum_{k=1}^{N_{\rm bit}^{\rm prbs}}
	\log_2 \big(1+\gamma_{k,i}\big)
	\ \text{[bits/sec]}.
\end{align}
With chip-rate bandwidth $B \approx 1/T_{\rm chip}$, the corresponding net spectral efficiency is $\bar{\mathcal{T}}_{\rm net} = \bar{\mathcal{C}}_{\rm sec}/B = 1/(N_{\rm prbs} N_{\rm chip}) \sum_{i=1}^{N_{\rm prbs}} \sum_{k=1}^{N_{\rm bit}^{\rm prbs}} \log_2(1+\gamma_{k,i})$~[bits/s/Hz].

Fig.~\ref{fig_cap_nsym} evaluates \eqref{eq_cap_seq_avg}--\eqref{eq_cap_net}: both $\bar{\mathcal{C}}_{\rm seq}$ and $\bar{\mathcal{C}}_{\rm sec}$ increase monotonically with $N_{\rm bit}^{\rm prbs}$, as the added parallel channel uses outweigh the lower per-slot SINR from the reduced spreading gain $L$, and the analytical curves (solid) coincide with the MC markers over the whole sweep. Together with the BER growth of Fig.~\ref{fig_ber_nsym}, this quantifies the throughput--reliability trade-off of the PDISAC waveform: dense allocations raise the rate at the cost of per-bit reliability, converting part of the radar-centric sounding overhead into a data rate that scales with $N_{\rm bit}^{\rm prbs}$, while Fig.~\ref{fig_rmse} confirms that RDPDNet contains the resulting sensing cost.

\section{Conclusions} \label{sec_conclusion}
\label{anchor_comment_2_4}
\label{anchor_comment_E_2c}
\label{anchor_comment_E_10d}
\label{anchor_comment_2_2c}
\label{anchor_comment_3_6d}
\label{anchor_comment_4_9d}

In this paper, we proposed PDISAC, an end-to-end framework for monostatic PMCW ISAC in stochastic, cluttered, and mobile environments, built on two coupled contributions: a multi-bit slot-partitioned ISAC waveform that embeds $N_{\rm bit}^{\rm prbs} \ge 1$ bits per sounding sequence through symbol-level spreading while preserving the deterministic radar code on every chip, and RDPDNet, a probabilistic denoising network trained with the AFM loss that suppresses the resulting  induced sidelobe artifacts without knowledge of the embedded symbols. We characterized the ensemble statistics of the geometry-determined channel and derived the closed-form CRLB for sensing together with the semi-analytical BER and average capacity for communication. We showed that RDPDNet absorbs most of the data-embedding sensing penalty and provides consistent low-SNR gains, that the estimator chains stay consistent with the bias-adjusted CRLB attained by the conventional data-free chain at high SNR, and that increasing the slot allocation raises the achievable rate at the expense of a higher BER, exposing a tunable sensing--communication trade-off; the denoising stage adds only a bounded, additive per-frame cost and preserves real-time capability.

\appendices

\section{Proof of Proposition 1}
\label{appx_proof_crlb}

From Section~\ref {sub_sec_crlb}, the FIM function can be rewritten as
\begin{align}
  \mathbf{F}(\vec{\Psi}) = \frac{2}{\sigma_{\rm sen}^{2}} & \textstyle \sum_{i=0}^{N_{\rm fr}-1} \Re\Bigg[
      \frac{\partial \sum_{m=1}^{N_{\rm tars}} \alpha_{\rm sen}^{(m)} r_{\mathrm{sen},i}^{(m)}(\vec{\psi}^{(m)})}{\partial \vec{\Psi}} \nonumber \\
      & \textstyle \times \left(\frac{\partial \sum_{m=1}^{N_{\rm tars}} \alpha_{\rm sen}^{(m)} r_{\mathrm{sen},i}^{(m)}(\vec{\psi}^{(m)})}{\partial \vec{\Psi}}\right)^{\mathsf{H}}
    \Bigg],
\end{align}
where the Hermitian product is required for real parameters in a complex Gaussian mean model, i.e., $F_{ab}=\frac{2}{\sigma_{\rm sen}^{2}}\sum_i\Re\{(\partial\mu_i/\partial\theta_a)^*(\partial\mu_i/\partial\theta_b)\}$, so that the complex gain contributes through $|\alpha_{\rm sen}^{(m)}|^2$.
Under the assumption that targets are orthogonal in the delay-Doppler
domain, the cross-target terms vanish, $\forall  m \neq k,$
\begin{align}
\hspace{-0.7em}\textstyle \sum_{i=0}^{N_{\rm fr}-1}\Re \left[\frac{\partial\alpha_{\rm sen}^{(m)}  r_{\mathrm{sen},i}^{(m)}(\vec{\psi}^{(m)})}{\partial\vec{\psi}^{(m)}}\left(\frac{\partial\alpha_{\rm sen}^{(k)}  r_{\mathrm{sen},i}^{(k)}(\vec{\psi}^{(k)})}{\partial\vec{\psi}^{(k)}}\right)^{\mathsf{H}}\right] \approx \vec{0},
\end{align}
and the joint FIM reduces to the block-diagonal form
$
  \mathbf{F}(\vec{\Psi})
  = \mathrm{diag} \left(
      \mathbf{F}(\vec{\psi}^{(1)}),
      \mathbf{F}(\vec{\psi}^{(2)}),
      \ldots, \mathbf{F}(\vec{\psi}^{(N_{\rm tars})})
    \right),
$
where each block is presented as
\begin{align}
  \mathbf{F} \left(\vec{\psi}^{(m)}\right)
  = \textstyle \frac{2|\alpha_{\rm sen}^{(m)}|^2}{\sigma_{\rm sen}^{2}}
    &\sum_{i=0}^{N_{\rm fr}-1}
    \Re \Bigg[
      \frac{\partial r_{\mathrm{sen},i}^{(m)}(\vec{\psi}^{(m)})}{\partial\vec{\psi}^{(m)}} \nonumber \\
    & \textstyle \times \left( 
        \frac{\partial r_{\mathrm{sen},i}^{(m)}(\vec{\psi}^{(m)})}{\partial\vec{\psi}^{(m)}}
       \right)^{ H}
    \Bigg].
\end{align}
Then, we derivatives with respect to range $r_m$ and velocity $v_m$ of the target $m$, with
$
  \frac{\partial r_{\mathrm{sen},i}^{(m)}}{\partial r_m}
  = \left(-\frac{2}{r_m} - j\frac{4\pi f_c}{c}\right)
    r_{\mathrm{sen},i}^{(m)} \label{eq_dr} 
  \quad - \frac{2}{c} \sqrt{\mathrm{PL}_{\mathrm{sen}}^{(m)}}
    e^{-j2\pi f_c \tau_{\mathrm{sen}}^{(m)}}
    e^{j2\pi f_{\mathrm{D},\mathrm{sen}}^{(m)} i T_{\rm chip}}
    d_i'(\tau_{\mathrm{sen}}^{(m)})
 $, and 
 $ 
  \frac{\partial r_{\mathrm{sen},i}^{(m)}}{\partial v_m}
  = \left(-j\frac{4\pi f_c i T_{\mathrm{chip}}}{c}\right)
    r_{\mathrm{sen},i}^{(m)}. \label{eq_dv}
$
, respectively.
From \eqref{eq_d_isac_t}, we have $d_i(t) \in \{-1+0j, 1+0j\}$ piecewise constant within each chip, so we set $d_i'(\tau_{\mathrm{sen}}^{(m)}) \approx 0$. This approximation neglects the bandwidth-dependent delay information carried by the chip transitions of the pulse-shaped waveform; together with the assumptions of a known complex gain $\alpha_{\rm sen}^{(m)}$ and orthogonal targets, the resulting bound is therefore conditional on these assumptions and optimistic for the delay. Substituting above steps into the target FIM, the three independent entries are simplified to
\begin{align}
  F_{rr}^{(m)}
  & =  \textstyle \frac{2\bigl|\alpha_{\mathrm{sen}}^{(m)}\bigr|^{2}
           N_{\rm fr} \mathrm{PL}_{\mathrm{sen}}^{(m)}}
         {\sigma_{\rm sen}^{2}}
    \left(\frac{4}{r_m^2}+\frac{16\pi^2}{\lambda^2}\right),
    \label{eq_Frr}\\[6pt]
 F_{rv}^{(m)}
  & =  \textstyle \frac{16\pi^2
           \bigl|\alpha_{\mathrm{sen}}^{(m)}\bigr|^{2}
           \mathrm{PL}_{\mathrm{sen}}^{(m)} T_{\mathrm{chip}} N_{\rm fr}(N_{\rm fr}-1)}
         {\sigma_{\rm sen}^{2}\lambda^2},
    \label{eq_Frv}\\[6pt]
  F_{vv}^{(m)}
  &= \textstyle \frac{16\pi^2
           \bigl|\alpha_{\mathrm{sen}}^{(m)}\bigr|^{2}
           \mathrm{PL}_{\mathrm{sen}}^{(m)}T_{\mathrm{chip}}^{2} N_{\rm fr}(N_{\rm fr} - 1)(2N_{\rm fr} - 1)}
         {3 \sigma_{\rm sen}^{2}\lambda^2}.
    \label{eq_Fvv}
\end{align}
The CRLB for the target $m$ follows directly from $\mathbf{F}(\vec{\psi}^{(m)})$, and it can be presented as
\begin{align}
  \hspace{-0.5em} \left[\mathbf{F} \left(\vec{\psi}^{(m)}\right)\right]^{-1}
   =  \frac{1}{\det \left(\mathbf{F}(\vec{\psi}^{(m)})\right)}
    \begin{bmatrix}
       F_{vv}^{(m)} & -F_{rv}^{(m)} \\[3pt]
      -F_{vr}^{(m)} &  F_{rr}^{(m)}
    \end{bmatrix}.
\end{align}
This completes the proof.

\section{Proof of Proposition~\ref{proposition_ber_b}}
\label{appx_proof_BER}
From the instantaneous error probability~\eqref{eq_ins_p_e_b} and the univariate change-of-variables theorem, the transformed probability density function $f_{p_e}(p_e)$ satisfies
\begin{align} \label{eq_cov_general}
    f_{p_e}(p_e)
    = f_{P_{\mathrm{nlos}}} \bigl(\Xi(p_e)\bigr)
      \times \big|\frac{dP_{\mathrm{nlos}}}{dp_e}\big|,
\end{align}
where $f_{P_{\mathrm{nlos}}}(\cdot)$ is the Gamma distribution evaluated for $N_{\mathrm{tars}}^{\mathrm{ref}}$ i.i.d.\ reflectors, and $\Xi(p_e)$ is the inverse mapping that expresses $P_{\mathrm{nlos}}$ as a function of $p_e$. We isolate the random aggregate interference $P_{\mathrm{nlos}}$ from the forward relationship~\eqref{eq_ins_p_e_b}. Applying the inverse $Q$-function and squaring both sides gives
$
    \bigl[Q^{-1}(p_e)\bigr]^2
    = \frac{2 P_{\mathrm{los}} L^2}
           {P_{\mathrm{nlos}} + \sigma_{\mathrm{com}}^2 L}.
$
Cross-multiplying and subtracting the thermal noise term yields the inverse mapping
$
    P_{\mathrm{nlos}}
    = \frac{2 P_{\mathrm{los}} L^2}
           {\bigl[Q^{-1}(p_e)\bigr]^2}
    - \sigma_{\mathrm{com}}^2 L
    \;\triangleq\; \Xi(p_e).
$
Note that $\Xi(p_e)\ge 0$ is guaranteed for all $p_e\in[p_e^{\min},1/2)$ by the choice of $p_e^{\min}$ in Proposition~\ref{proposition_ber_b}, since at $p_e = p_e^{\min}$ the interference term vanishes ($P_{\mathrm{nlos}}=0$) and $\Xi(p_e)$ increases monotonically toward infinity as $p_e \to 1/2$, so the domain covers the full support $P_{\mathrm{nlos}}\in[0,\infty)$ of the Gamma-distributed interference.

Consider the Jacobian computation, we compute $|dP_{\mathrm{nlos}}/dp_e|$ via the chain rule, introducing the intermediate variable $w = Q^{-1}(p_e)$ as
$
    \frac{dP_{\mathrm{nlos}}}{dp_e}
    = \frac{dP_{\mathrm{nlos}}}{dw}  \frac{dw}{dp_e}.
$
Differentiating $\Xi$ with respect to $w$,
$
    \frac{dP_{\mathrm{nlos}}}{dw}
    = -\frac{4 P_{\mathrm{los}} L^2}{w^3}.
$
Differentiating $p_e = Q(w)$ using Leibniz's integral rule,
$
    \frac{dp_e}{dw}
    = -\frac{1}{\sqrt{2\pi}} e^{-w^2/2}
    \rightarrow
    \frac{dw}{dp_e}
    = -\sqrt{2\pi} \exp \big(\frac{w^2}{2}\big).
$
Substituting $w = Q^{-1}(p_e)$ and multiplying the two factors gives the absolute Jacobian
$
    \big|\frac{dP_{\mathrm{nlos}}}{dp_e}\big|
    = \frac{4\sqrt{2\pi} P_{\mathrm{los}} L^2}
           {\bigl[Q^{-1}(p_e)\bigr]^3}
      \exp \big(\frac{\bigl[Q^{-1}(p_e)\bigr]^2}{2}\big).
$

Consider the Gamma density $f_{P_{\mathrm{nlos}}}(x) = \frac{x^{N_{\mathrm{tars}}^{\mathrm{ref}}-1} e^{-x/\Omega}}{\Omega^{N_{\mathrm{tars}}^{\mathrm{ref}}} (N_{\mathrm{tars}}^{\mathrm{ref}}-1)!}$ evaluated at the mapped coordinate $x = \Xi(p_e)$; combining with the absolute Jacobian via~\eqref{eq_cov_general} yields the closed-form instantaneous PDF
\begin{align} \label{eq_pdf_Pe}
    &f_{p_e}(p_e)
    = f_{P_{\mathrm{nlos}}} \bigl(\Xi(p_e)\bigr)
        \big|\frac{dP_{\mathrm{nlos}}}{dp_e}\big| \\[4pt]
    &= \textstyle \frac{\bigl[\Xi(p_e)\bigr]^{N_{\mathrm{tars}}^{\mathrm{ref}}-1}
             \exp \big(-\dfrac{\Xi(p_e)}{\Omega}\big)}
            {\Omega^{N_{\mathrm{tars}}^{\mathrm{ref}}}
             \bigl(N_{\mathrm{tars}}^{\mathrm{ref}}-1\bigr)!}
      \frac{4\sqrt{2\pi} P_{\mathrm{los}} L^2}
           {\bigl[Q^{-1}(p_e)\bigr]^3}
      \exp \big(\frac{\bigl[Q^{-1}(p_e)\bigr]^2}{2}\big),  \nonumber
\end{align}
where $\Xi(p_e) = \frac{2 P_{\mathrm{los}} L^2}{\bigl[Q^{-1}(p_e)\bigr]^2} - \sigma_{\mathrm{com}}^2 L$. Applying $\bar{P}_e = \int_{p_e^{\min}}^{1/2} p_e f_{p_e}(p_e) dp_e$ directly yields the expected average BER formulation in \eqref{eq_avg_ber_b}. This completes the proof.

\color{black}
\bibliographystyle{IEEEtran}

\bibliography{References}
\end{document}